\pdfoutput=1

\documentclass[11pt]{article}

\usepackage{fullpage}
\usepackage[english]{babel}

\usepackage{enumerate}

\usepackage{amsmath}
\usepackage{amsthm}
\usepackage{thmtools} 
\usepackage{thm-restate}
\usepackage{amssymb}
\usepackage{mathtools}
\usepackage{nicefrac}
\usepackage{bm}
\usepackage{bbm}
\usepackage{multirow}
\usepackage{diagbox}
\usepackage{graphicx}
\usepackage{subfig}
\usepackage[hidelinks,pdfusetitle,pdfencoding=auto,psdextra,bookmarksopen]{hyperref}
\usepackage[textsize=tiny,backgroundcolor=white,linecolor=black]{todonotes}
\usepackage{tikz}
\usepackage[normalem]{ulem}
\usepackage{mathabx}
\usepackage{cleveref}
\usepackage{optidef}
\usepackage[procnumbered,ruled,vlined,linesnumbered, algo2e]{algorithm2e}
\usepackage{comment}

\usepackage{amsthm}

\let\oldnl\nl
\newcommand{\nonl}{\renewcommand{\nl}{\let\nl\oldnl}}

 \newtheorem{@theorem}{Theorem}[section]
 \newenvironment{theorem}{\begin{@theorem}}{\end{@theorem}}

\newtheorem{lemma}{Lemma}[section]

\newtheorem{corollary}{Corollary}[section]

\newtheorem{Definition}{Definition}[section]

\newtheorem{observation}{Observation}[section]
\newtheorem{claim}{Claim}[section]

\def\eps{\ensuremath{\varepsilon}}
\def\OPT{\ensuremath{\mathrm{OPT}}}

\def\NNCC{Nested Norm $k$-Clustering}

\newcommand{\I}[0]{\mathcal{I}}

\def\PZ{ball $k$-median}

\def\MSRDC{Min-Sum of Radius-Dependent Cost}
\def\MSRDCS{MSRDC}
\def\FLMSRDCS{FL-MSRDC}

\def\Ballk{Ball $k$-Median}
\def\FLBall{Ball Facility Location}
\def\msr{Min-Sum Radii}
\def\kmed{$k$-Median}

\newcommand\LP[1]{\mathcal{L}_{#1}}

\def\X{\ensuremath{\mathcal{X}}}

\def\eins{\ensuremath{\mathbbm{1}}}
\def\tildw{\ensuremath{\bm{\Tilde{w}}}}
\def\tildx{\ensuremath{\Tilde{x}}}
\def\nnr{\ensuremath{\mathbb{R}_{\ge0}}}
\def\nnrvec{\ensuremath{\mathbb{R}^n_{\ge 0}}}

\newcommand\dif[1]{\textnormal{\ensuremath{\textsf{dif}(#1)}}}
\newcommand\topl[2]{\textnormal{\ensuremath{\textsf{top}_#1(#2)}}}

\newcommand\ord[2]{\textnormal{\ensuremath{\textsf{ord}_{#1}(#2)}}}

\newcommand\dist[2]{\ensuremath{\delta(#1,#2)}}
\newcommand\ldist[3]{\ensuremath{\delta^{\downarrow}_{#3}(#2,#1)}}
\newcommand\distr[3]{\ensuremath{\delta^{#3}(#1,#2)}}

\newcommand\distv[2]{\ensuremath{\bm{\delta}_{#1}(#2)}}

\newcommand\dists[2]{\ensuremath{\delta^*(#1,#2)}}
\newcommand\distss[2]{\ensuremath{\delta^{**}(#1,#2)}}

\newcommand\cost[2]{\textnormal{\ensuremath{\textsf{cost}_{#1}(#2)}}}

\newcommand\costz[1]{\textnormal{\ensuremath{\textsf{cost}_{\textsf{b}}(#1)}}}
\newcommand\costd[2]{\textnormal{\ensuremath{\textsf{cost}_{\textsf{MaxOrd}}(#1,#2)}}}

\newcommand\proxyz[3]{\textnormal{\ensuremath{\textsf{proxy}_{#2}(#3,#1)}}}
\newcommand\proxy[3]{\textnormal{\ensuremath{\textsf{proxy}_{#3}(#1,#2)}}}

\newcommand\NCCH[2]{Nested \ensuremath{(#1,#2)} k-Clustering}
\newcommand\NCCS[2]{\ensuremath{(#1,#2)}-Clustering}
\newcommand\NCC[2]{\NCCS{#1}{#2}}

\newcommand{\set}[1]{\{#1\}}

\def\sqsubsetneq{\mathrel{\sqsubseteq\kern-0.92em\raise-0.15em\hbox{\rotatebox{313}{\scalebox{1.1}[0.75]{\(\shortmid\)}}}\scalebox{0.3}[1]{\ }}}
\def\sqsupsetneq{\mathrel{\sqsupseteq\kern-0.92em\raise-0.15em\hbox{\rotatebox{313}{\scalebox{1.1}[0.75]{\(\shortmid\)}}}\scalebox{0.3}[1]{\ }}}

\renewcommand{\subparagraph}{\paragraph}

\newcommand{\floor}[1]{\ensuremath{\left\lfloor #1 \right\rfloor}}
\newcommand{\ceil}[1]{\ensuremath{\left\lceil #1 \right\rceil}}

\newcommand{\orc}{\includegraphics[height=\fontcharht\font`A]{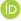}}

\title{\Large Clustering to Minimize Cluster-Aware Norm Objectives\footnote{Martin Herold is funded by the Deutsche Forschungsgemeinschaft (DFG, German Research Foundation) – Project number 399223600.
This work is part of the project TIPEA that has received funding from the European Research Council (ERC)
under the European Union's Horizon 2020 research and innovation programme (grant agreement No. 850979).
We are grateful to an anonymous reviewer for making concrete suggestions how to substantially simplify our algorithm for \NCCS{\textsf{Top}}{\LP{1}} by applying the techniques of Ahmadian, Swamy~\cite{ahmadian}.}}
\author{Martin G.\ Herold\href{https://orcid.org/0009-0002-1804-2842}{\orc}\thanks{Max Planck Institute for Informatics, Saarland Informatics Campus, Germany.} \and Evangelos Kipouridis\href{https://orcid.org/0000-0002-5830-5830}{\orc}\thanks{Saarland University and Max Planck Institute for Informatics, Saarland Informatics Campus, Germany.} \and Joachim Spoerhase\href{https://orcid.org/0000-0002-2601-6452}{\orc}\thanks{University of Liverpool, United Kingdom.}}

\begin{document}
\frenchspacing
\maketitle

\begin{abstract}
We initiate the study of the following general clustering problem. We seek to partition a given set $P$ of data points into $k$ clusters by finding a set $X$ of $k$ centers and assigning each data point to one of the centers. The cost of a cluster, represented by a center $x\in X$, is a monotone, symmetric norm $f$ (called inner norm) of the vector of distances of points assigned to $x$. The  goal is to minimize a norm $g$ (called outer norm) of the vector of cluster costs. This problem, which we call \NCCS{f}{g}, generalizes many fundamental clustering problems such as $k$-Center (i.e., \NCCS{\LP{\infty}}{\LP{\infty}}), $k$-Median (i.e., \NCCS{\LP{1}}{\LP{1}}), Min-Sum of Radii (i.e., \NCCS{\LP{\infty}}{\LP{1}}), and Min-Load $k$-Clustering (i.e., \NCCS{\LP{1}}{\LP{\infty}}). A recent line of research (Byrka et al. [STOC'18], Chakrabarty, Swamy [ICALP'18, STOC'19], and Abbasi et al. [FOCS'23]) studies norm objectives that are oblivious to the cluster structure such as $k$-Median and $k$-Center. In contrast, our problem models cluster-aware objectives including Min-Sum of Radii and Min-Load $k$-Clustering.

Our main results are as follows. First, we design a constant-factor approximation algorithm for \NCCS{\textsf{top}_\ell}{\LP{1}} where the inner norm ($\textsf{top}_\ell$) sums over the $\ell$ largest distances. This unifies (up to constant factors) the best known results for $k$-Median and Min-Sum of Radii. Second, we design a constant-factor approximation\ for \NCCS{\LP{\infty}}{\textsf{Ord}} where the outer norm is a convex combination of $\textsf{top}_\ell$ norms (ordered weighted norm). This generalizes known results for $k$-Center and Min-Sum of Radii. Obtaining a constant-factor approximation for more general settings that include \NCCS{\LP{1}}{\LP{\infty}} (Min-Load $k$-Clustering) seems challenging because even an $o(k)$-approximation is unknown for this problem. We can still use our two main results to obtain first (although non-constant) approximations for these problems including general monotone, symmetric norms.

Our algorithm for \NCCS{\textsf{top}_\ell}{\LP{1}} relies on a reduction to a novel generalization of $k$-Median, which we call Ball $k$-Median. In this problem, we aim at selecting $k$ balls (rather than $k$ centers) and pay for connecting the points to these balls as well as for the (scaled) radii of the balls. To obtain a constant-factor approximation for this problem we unify various algorithmic techniques originally designed for the cluster-oblivious $k$-Median objective (Jain and Vazirani [JACM 2001], Li and Svensson [STOC'13]) and for the cluster-aware \msr{} Objective (Charikar and Panigrahi [STOC'01] and Ahmadian and Swamy [ICALP'16]).

\end{abstract}

\newcommand\x{0}

\tableofcontents
\newpage
\setcounter{page}{1} 
\section{Introduction}

Clustering is among the most fundamental tasks in data analysis, computer science, and operations research. It concerns partitioning a set~$P$ of data points from a metric space $M$ into $k$ groups (clusters) of points that are close to each other. We propose the following general class of clustering problems. A clustering is specified by a pair $(X,\sigma)$ where $X$ is a $k$-element subset of a set~$F$ of potential cluster centers from~$M$, and where $\sigma\colon P\rightarrow X$ assigns each data point to a center. We associate each cluster center $x\in X$ with a distance vector $\distv{\sigma}{x}=\left(\delta(p,x)\mathbbm{1}[x=\sigma(p)]\right)_{p\in P}$ where~$\delta$ denotes the distance function in~$M$. Given a monotone, symmetric norm~$f$ (which we call inner norm), we assign to each cluster center $x$ a cluster cost $f_\sigma(x):=f(\distv{\sigma}{x})$. Our goal is to find a solution $(X,\sigma)$ minimizing $\cost{\sigma}{X}: = g\left(\left( f(\distv{\sigma}{x} \right)_{x\in X}\right)$ where $g$ is another monotone, symmetric norm $g$  (which we call outer norm).  We call this problem \NCCS{f}{g}. (See also Definitions~\ref{def:nncc} and~\ref{def:ncc}.) 

This model is quite versatile as it generalizes fundamental clustering problems such as $k$-Center (i.e., \NCCS{\LP{\infty}}{\LP{\infty}}), $k$-Median (i.e., \NCCS{\LP{1}}{\LP{1}}), Min-Sum of Radii (i.e., \NCCS{\LP{\infty}}{\LP{1}}), Min-Load $k$-Clustering (i.e., \NCCS{\LP{1}}{\LP{\infty}}), $k$-Means (i.e., \NCCS{\LP{2}}{\LP{2}}), and the more general $(k,z)$-Clustering problem (i.e., \NCCS{\LP{z}}{\LP{z}}). Informally, these problems have in common that the objective aggregates (via the outer norm) over suitably defined cluster costs (via the inner norm), which is why we call them \emph{cluster-aware}. This is in contrast to \emph{cluster-oblivious} problems where the objective is a function of the (global) distance vector $(\delta(p,\sigma(p)))_{p\in P}$. Notice that $k$-Median and $k$-Center are contained in both classes whereas Min-Sum of Radii and Min-Load $k$-Clustering are cluster-aware but not cluster-oblivious. 

There has been a recent rise of interest in more general norm objectives in various areas such as (cluster-oblivious) clustering~\cite{joachim,otherOrderedKMedian,aouad-segev19ordered-k-median,chakrabarty-swamy19:norm-k-clustering,chlamtac-etal22:fair-cascaded-norm-clustering,abbasi-etal23:epas-norm-clustering}, load balancing~\cite{chakrabarty-swamy19:norm-k-clustering,deng-etal23:Generalized-Load-Balancing}, and stochastic optimization~\cite{ibrahimpur-swamy20:stochastic-norm-optimization}. The algorithmic study of such generalizations helps unify algorithmic techniques. They are of particular importance in clustering due to its diverse range of applications with often poorly characterized objectives. Additionally, they lead to new objectives interpolating between the classic objectives. 

Cluster-aware objectives are important for a variety of reasons. For example, it has been observed that cluster-oblivious objectives may lead to dissection of natural clusters (see \Cref{fig:dissection}). In fact, the Min-Sum of Radii problem has been suggested as a cluster-aware objective that reduces such dissection effects~\cite{hansen-jaumard97cluster-analysis,MinSumRadii}. In Section~\ref{sec:motivating-example}, we discuss this dissection effect and a scenario motivating objectives in-between \msr{} and $k$-Median. While there has been a recent line of research on cluster-oblivious norm objectives~\cite{joachim,otherOrderedKMedian,aouad-segev19ordered-k-median,chakrabarty-swamy19:norm-k-clustering,chlamtac-etal22:fair-cascaded-norm-clustering,abbasi-etal23:epas-norm-clustering} such a framework is lacking for cluster-aware objectives. We aim at bridging this gap with our work.

A notable structural difference between optimal solutions to cluster-oblivious and to cluster-aware objectives is the following. Under cluster-oblivious objectives data points are w.l.o.g.\ assigned to the nearest cluster center; this is in fact used by many algorithms for such objectives.
In contrast, cluster-aware objectives can incentivize solutions to assign data points to more distant but more suitable centers (see \Cref{fig:dissection}). This makes cluster-aware objectives preferable in certain applications but potentially harder to handle algorithmically.

\paragraph{State of the Art.} We briefly outline the state of the art for concrete cluster-aware objectives as well as related lines of research. For more details, we refer to Section~\ref{sec:sota}. Most natural cluster-aware problems are NP-hard~\cite{hochbaum-shmoys85:k-center,guha-khuller99:greedy-facility-location,gibson-etal10:msr,ahmadian-etal18:min-load-k-median}, which inspired intensive research on approximation algorithms for these problems~\cite{hochbaum-shmoys85:k-center,JainVaz,MinSumRadii,jain-etal03:greedy-facility-location,arya-etal04:local-search-k-median,ola,ahmadian,ahmadian-etal18:min-load-k-median,friggstad-jamshidian22:msr}. For Min-Sum of Radii, $k$-Center, and $k$-Median, the best known approximation algorithms have a constant guarantee~\cite{buchem-etal24:msr,hochbaum-shmoys85:k-center,byrka-etal17:improved-k-median}. The best approximation Min-Load $k$-Clustering has ratio $O(k)$ and improving this to $o(k)$ is elusive~\cite{ahmadian-etal18:min-load-k-median}.

Clustering-oblivious norm objectives admit an $O(1)$-approximation~\cite{chakrabarty-swamy19:norm-k-clustering} substantially generalizing the results for $k$-Median and $k$-Center. Crucial intermediate steps were obtaining $O(1)$-approximations for $\textsf{top}_{\ell}$ norms, which sum over the $\ell$ largest coordinates, and ordered weighted norms, which are convex combinations of the $\textsf{top}_{\ell}$ norms~\cite{joachim,otherOrderedKMedian,aouad-segev19ordered-k-median}.

A line of research related to ours concerns allocation tasks under generalized objectives. Deng et al.~\cite{deng-etal23:Generalized-Load-Balancing} introduce \emph{Generalized Load Balancing}, in which we wish to assign each job $j$ to a machine $i$ incurring a processing time $p(i,j)$. 
The \emph{load} of a machine $i$ is a symmetric, monotone (inner) norm of the vector of processing times of jobs assigned to $i$. 
The overall objective is another (outer) norm aggregating the loads, which we want to minimize. Their setting has some parallels to ours once we fix the set of opened centers. Their main result is an $O(\log n)$-approximation, which is best possible unless $\textsf{P}=\textsf{NP}$ even for the special case where the outer norm is $\LP{1}$ and the inner norm is $\LP{\infty}$. 
This is related to the fact that their processing times do not satisfy the triangle inequality. 
In their setting the machines are fixed while a crucial difficulty in clustering lies in selecting the $k$ centers. 
Thus, there seems to be no obvious way to use their results to obtain $O(1)$-approximations for our settings. 
Other works~\cite{svitkina-tardos10:hierarchical-facility-location,svitkina-fleischer11:submodular-approx,chekuri-ene11:submodular-cost-allocation,abbasi-etal24:submodular-fac-loc} study allocation tasks similar to generalized load balancing but use a submodular function instead of the inner norm and restrict the outer norms to be $\LP{1}$ or $\LP{\infty}$. 
For none of these problems $O(1)$-approximations are known and most of them are provably $O(\log n)$-hard to approximate. Deng et al.~\cite{deng-etal23:Generalized-Load-Balancing} show that in these settings submodular functions do not capture monotone, symmetric norms.

\subsection{Our Results}
In the following, we use \textsf{Top} to denote the class of $\topl{\ell}{\cdot}$-norms, \textsf{Ord} to denote the class of ordered norms, and \textsf{Sym} to denote the class of symmetric monotone norms (see \Cref{ssec:typesOfNorms} for precise definitions).

Our two main results are the following. 
First, we provide an $O(1)$-approximation algorithm for \NCCS{\textsf{Top}}{\LP{1}}. This unifies (up to constant factors), the best known approximations for $k$-Median and \msr{} and allows us to interpolate between these two classic problems.
\begin{restatable}{theorem}{apxtoplone}
\label{thm:apxtoplone}
    For all $\eps>0$, there is a factor-$(13.5+\eps)$ approximation for \NCC{\textnormal{\textsf{Top}}}{ \LP{1} }. 
\end{restatable}

Second, we unify the $O(1)$-approximations for \msr{} and $k$-Center and interpolate between these to problems, by showing an $O(1)$-approximation for \NCCS{\LP{\infty}}{\textsf{Ord}}.

\begin{restatable}{theorem}{thmlinford}\label{thm:linford}
    For all $\eps>0$ there is a factor-$(18+\eps)$ approximation for \NCCS{\LP{\infty}}{\textnormal{\textsf{Ord}}}.
\end{restatable}

Combining these two results with general properties of the different types of norms we obtain first (although non-constant) approximation algorithms for many further settings. 
These results are summarized in \Cref{tbl:results}. 
We remark that many of the natural generalizations contain Min-Load $k$-Clustering as a special case for which an $o(k)$-approximation is elusive~\cite{ahmadian-etal18:min-load-k-median}. 
Notice that our two main results provide a contrast to the related Generalized Load Balancing problem~\cite{deng-etal23:Generalized-Load-Balancing}, which is $O(\log n)$-hard to approximate even when the outer norm is $\LP{1}$ and the inner norm is $\LP{\infty}$.
In \Cref{thm:bicriteria}, we design a bicriteria $(O(1),O(1))$-approximation for \NCCS{\LP{\infty}}{\textsf{Sym}}, that is a constant factor approximation opening $O(k)$ (instead of $k$) facilities.
This is an independent result not obtained from the two main results. We obtain this result by identifying a connection to the Non-Uniform $k$-Center problem~\cite{NUkC}.

It is an interesting open question if there are $O(1)$-approximations for \NCCS{\textsf{Ord}}{\LP{1}} or \NCCS{\LP{\infty}}{\textsf{Sym}}. Notice that for cluster-oblivious objectives $\textsf{top}_{\ell}$ norms have previously served as a crucial intermediate step to handle ordered weighted norms~\cite{joachim,otherOrderedKMedian,aouad-segev19ordered-k-median}. 
Similarly, handling general norms crucially relies on ordered weighted norms~\cite{chakrabarty-swamy19:norm-k-clustering}. 
We therefore believe that our two main results may help to resolve these open questions.

\begin{table}[!ht]
\caption{Best-known approximation factors for several important cases of \NNCC{}. The bold entries are the main results of this paper. Entries without a reference are stated in \Cref{cor:listofcors}.}
\label{tbl:results}
\begin{center}
    \begin{tabular}{ | c | c | c | c | c |}
    \hline
    \diagbox{Outer}{Inner}& $\LP{1}$ & $\LP{\infty}$ & \textsf{Top} & \textsf{Sym}\\ \hline
    \multirow{2}{6em}{\centering$\LP{1}$} & {$O(1)$}& {$O(1)$}& {$\mathbf{O(1)}$}& \multirow{2}{6em}{\centering$\widetilde{O}(\sqrt{n})$}\\ 
    &\cite{JainVaz}&\cite{MinSumRadii}&\textbf{\Cref{thm:apxtoplone}}& \\\hline
    \multirow{2}{6em}{\centering$\LP{\infty}$} & {$O(k)$}& {$O(1)$}& \multirow{4}{6em}{\centering $O(k)$} & \multirow{4}{6em}{\centering$\widetilde{O}(\sqrt{nk})$}\\
    &\cite{ahmadian-etal18:min-load-k-median}&\cite{hochbaum-shmoys85:k-center,gonzalez85:k-center}& &\\\cline{1-3}
    \multirow{2}{6em}{\centering \textsf{Ord}} & \multirow{2}{6em}{\centering$O\left(k\right)$}& \multirow{2}{6em}{\centering$\mathbf{O(1)}$\\ \centering\textbf{\Cref{thm:linford}}}& &\\
    & & &  &\\ \hline
    \multirow{3}{6em}{\centering \textsf{Sym}} & \multirow{3}{6em}{\centering$O(k\log k)$}& \begin{tabular}{@{}c@{}}$\mathbf{(O(1),O(1))}$\\ \textbf{\Cref{thm:bicriteria}}\end{tabular}& \multirow{3}{6em}{\centering$O(k\log k)$} &\multirow{3}{8em}{\centering$\widetilde{O}(\sqrt{nk})$}\\    
    &&\begin{tabular}{@{}c@{}}\multirow{2}{6em}{\centering $O(\log k)$}\\ \hspace{2pt}\end{tabular}& & \\\hline
    \end{tabular}
\end{center}
\end{table}

\subsection{Our Techniques}

A known difficulty of handling norm objectives is their non-linearity. Works for cluster-oblivious problems~\cite{aouad-segev19ordered-k-median,joachim,otherOrderedKMedian,chakrabarty-swamy19:norm-k-clustering} therefore replace the norm objective with a \emph{proxy} objective. For example,  the $\topl{\ell}{\bm{x}}$ norm, that sums over the $\ell$ largest coordinates of $\bm{x}$, is replaced with $\proxy{y}{\bm{x}}{\ell} = \ell \cdot y  + \sum_{i=1}^n (x_i \dotdiv y)$
where $y$ is a threshold parameter and $a\dotdiv b=\max\{0,a-b\}$. It is used that $\topl{\ell}{\bm{x}}=\proxy{\bm{x}^{\downarrow}[\ell]}{\bm{x}}{\ell} $ where $\bm{x}^{\downarrow}[\ell]$ denotes the $\ell$-th largest coordinate in $\bm{x}$. 
Morever,  $\proxy{y}{\bm{x}}{\ell} \ge \topl{\ell}{\bm{x}}$ for all $y$. The benefit of the proxy function is that it is a linear function over the modified vector $(x_i\dotdiv y)_{i\in [n]}$. Two difficulties arise from this. First, in order to preserve the optimum objective value, we need to choose the threshold parameter $y$ to be the $\ell$-th largest distance $\bm{o}^{\downarrow}[\ell]$ in the optimum  $\bm{o}$ cost vector. The second difficulty is specific for clustering where the cost vector $\bm{x}$ represents a distance vector such as $(\delta(p,\sigma(p))_{p\in P}$ for cluster-oblivious objectives. Modifying this cost vector in clustering means to modify the distances by substracting the threshold, which may lead to violation of the triangle inequality. However, classic clustering algorithms crucially rely on the triangle inequality.

For cluster-oblivious objectives, the first difficulty of choosing the right threshold~$y$ can be handled by simply guessing~$\bm{o}^{\downarrow}[\ell]$. This idea fails, however, for handling the inner norm in cluster-aware objectives. We would have to guess for every optimal \emph{center} $o$ the $\ell$-largest distance $\bm{\delta}^{\downarrow}_\sigma(o)[\ell]$ in its cluster. This is computationally infeasible as there are $k$ clusters. We illustrate how we overcome this difficulty for \NCCS{\textsf{top}_\ell}{\LP{1}}. Instead of guessing the thresholds, we reduce the problem to the following new and natural variant of $k$-Median, which we call \emph{Ball $k$-Median}. In this problem, we want to select $k$ balls (specified by center and radius) rather than $k$ centers. We aim at minimizing the total cost of connecting the clients to the balls plus the (scaled) sum of radii (see Definition~\ref{def:ball-k-median}). Connecting a point $p$ to a ball with radius $r$ centered at $x$ costs $\delta^r(p,x)=\delta(p,x)\dotdiv r$, which violates the triangle inequality.

There are two key properties of Ball $k$-Median that we leverage in our main result: First, the radius selection of the balls encapsulates the choice of the thresholds in the proxy functions thereby bypassing the guessing used for cluster-oblivous objectives~\cite{joachim,otherOrderedKMedian}.
Second, Ball $k$-Median admits a natural LP relaxation that is amenable to applying the Langrange relaxation and bi-point rounding framework by Jain and Vazirani~\cite{JainVaz} and used by many algorithms for $k$-Median and Min-Sum of Radii.

This framework has many steps (outlined below) that fulfil, on a high-level, similar purposes for $k$-Median and Min-Sum of Radii. The concrete implementations of the steps are very different for both problems, and have previously been part separate lines of research. Our main insight is that, by generalizing, extending, and combining techniques from four various different works~\cite{JainVaz,ola,MinSumRadii,ahmadian}, all these steps can be unified into a single algorithm. This requires us to find the right generalization of these steps previously designed for $k$-Median and Min-Sum of Radii, separately. Additionally, as in previous works on cluster-oblivious objectives, dealing with the non-metric distances requires non-trivial ideas in most of these steps. Overall, this demonstrates a versatility of the techniques by~\cite{JainVaz,ola,MinSumRadii,ahmadian} that we find surprising.

The first step of this framework consists in a primal-dual approximation algorithm for the Lagrange relaxation of the clustering problem, where we do not restrict the number of centers but pay for opening them. For the resulting Ball Facility Location problem, we design a primal-dual algorithm unifying the algorithm by Jain and Vazirani~\cite{JainVaz} for classic Facility Location and by Charikar and Panigrahi~\cite{MinSumRadii}. Specifically, the dual ascent phase extends the one by~\cite{JainVaz} whereas the pruning phase builds on~\cite{MinSumRadii}.

The second step of the framework uses the primal-dual algorithm to compute a fractional \emph{bi-point solution} for Ball $k$-Median, which is a convex combination of two solutions $X_1$ and $X_2$ for Ball Facility Location where $|X_1|<k$ and $|X_2|>k$. The computation of the bi-point solution is straightforward for Ball $k$-Median. This bi-point solution is then rounded to an integral feasible solution for Ball $k$-Median. The bi-point rounding step poses the most technical challenges in our algorithm. We unify bi-point rounding approaches by Li and Svensson~\cite{ola} for $k$-Median and by Ahmadian and Swamy~\cite{ahmadian} for Min-Sum of Radii. Both approaches are based on setting up an auxiliary knapsack LP to partially round the bi-point solution to an integral solution apart from one fractional group, called \emph{special group}, consisting of one center from $X_1$ and multiple centers from $X_2$. Some of the bi-point rounding steps require notable changes of the strategy, for example, when handling the special group. See Sections~\ref{subsec:apxballk} and~\ref{subsubsec:constructingvialp}.

\subsection{Previous and Related Work}\label{sec:sota}
\paragraph{Cluster-Aware Objectives.}

Most of the natural clustering problems are NP-hard such as \msr{}~\cite{gibson-etal10:msr} or even APX-hard such as $k$-Median~\cite{guha-khuller99:greedy-facility-location} and $k$-Center~\cite{hochbaum-shmoys85:k-center,gonzalez85:k-center} and Min-Load $k$-Clustering~\cite{ahmadian-etal18:min-load-k-median}.

This inspired intensive research on approximation algorithms for these problems leading to the development of a rich toolbox of algorithmic techniques based on, for example, greedy, local search, primal-dual, or LP-rounding. There is a series of improved approximation algorithms for Min-Sum of Radii~\cite{MinSumRadii,friggstad-jamshidian22:msr} with the currently best approximation $3+\epsilon$ by Buchem et al.~\cite{buchem-etal24:msr}.
Interestingly enough, it admits a quasi-polynomial time approximation scheme~\cite{gibson-etal10:msr} and is therefore probably not APX-hard. For $k$-Center $2$-approximations are known, which is best possible unless $\textsf{P}=\textsf{NP}$~\cite{hochbaum-shmoys85:k-center,gonzalez85:k-center}. 
There is an intensive line of research improving the approximation factors for $k$-Median~\cite{charikar-etal02:constant-k-median,JainVaz,jain-etal03:greedy-facility-location,arya-etal04:local-search-k-median,ola}. The currently best approximation by Gowda et al.~\cite{gowdaetal2023:bestkmed} has a ratio of $2.613+\epsilon$. 
The Min-Load $k$-Clustering problem, in contrast, is much less understood. 
An $O(k)$-approximation follows from the $O(1)$-approximations for $k$-Median, and approximation schemes are known for line metrics~\cite{ahmadian-etal18:min-load-k-median}. 
However, an $o(k)$-approximation for general metrics is elusive.

\paragraph{Cluster-Oblivious Norm Objectives.}
There has been a recent interest in generalized objectives for cluster-oblivious problems. A first set of result focused on $\textsf{top}_{\ell}$ (called $\ell$-Centrum) and the more general ordered weighted objectives (called Ordered $k$-Median) obtaining logarithmic approximations~\cite{aouad-segev19ordered-k-median}. Byrka et al.~\cite{joachim}, and Chakrabarty and Swamy~\cite{otherOrderedKMedian} obtain the first constant-factor approximations for Ordered $k$-Median, which unifies constant-factor approximations for $k$-Median and $k$-Center and also implies a constant-factor approximation for $\ell$-Centrum. This line of research culminated in the constant-factor approximation for general (cluster-oblivious) monotone, symmetric norms by Chakrabarty and Swamy~\cite{chakrabarty-swamy19:norm-k-clustering} further generalizing ordered weighted norms.
Chlamt{\'{a}\v{c}} et al.~\cite{chlamtac-etal22:fair-cascaded-norm-clustering} study $(p,q)$-fair clustering where the data points are partitioned into groups (more generally described by multiple weight functions). Each group is assigned a cost under the $\LP{p}$ norm and the overall cost is the $\LP{q}$ norm of the group costs. While their clustering objective involves nested norms as well, their groups are fixed by the input whereas our ``groups'' (clusters) are to be determined as part of the solution. Also, we consider general monotone, symmetric norm and focus on $\textsf{top}_{\ell}$ and ordered weighted norms in particular rather than $\LP{p}$-norms objectives. Abbasi et al.~\cite{abbasi-etal23:epas-norm-clustering} study general \emph{asymmetric} monotone norms, which subsume all cluster-oblivious objectives described above (including $(p,q)$-fair clustering). They develop an efficient parameterized approximation scheme for structured metric spaces such as high-dimensional Euclidean space, bounded doubling metrics, and shortest path metrics in bounded tree-width and planar graphs.

Notice that all of the above problems are cluster-oblivious and therefore do not capture cluster-aware objectives such as Min-Sum of Radii and Min-Load $k$-Clustering.

\paragraph{Generalized Load Balancing.}
Our problem is related to the Generalized Load Balancing problem recently introduced by Deng et al.~\cite{deng-etal23:Generalized-Load-Balancing}. In this problem, we are given a set of jobs (related to our data points) and a set of machines (related to our facilities). Executing a job~$j$ on a machine~$i$ incurs a processing time~$p(i,j)$ (related to point-center distances). The load of machine $i$ is computed by a monotone, symmetric norm (inner norm) $\psi_i$ of the vector of processing times of jobs assigned to~$i$. The loads of the machines are then aggregated via an outer norm to the overall objective function~$\phi$, which we wish to minimize. Their main result is an $O(\log n)$-approximation algorithm for this problem, which is best possible unless $\textsf{P}=\textsf{NP}$. Notice that their setting is incomparable to ours. It does not capture the selection of a $k$-subset of centers because the set of machines is fixed. On the other hand, their inner norms are machine-specific and their processing times do not need to satisfy the triangle inequality. An earlier work by Chakrabarty and Swamy~\cite{chakrabarty-swamy19:norm-k-clustering} introduces the special of norm load balancing where the inner norm is $\LP{1}$ and obtain an $2$-approximation for it.

\paragraph{Submodular Load Balancing, Allocation, and Facility Location.} 
Svitkina and Fleischer~\cite{svitkina-fleischer11:submodular-approx} study the related setting of Submodular Load Balancing where replace the inner norm with a submodular function and use $\LP{\infty}$ as an outer norm. They obtain an $O(\sqrt{n/\log n})$-approximation for this problem along with matching lower bounds. If we use instead $\LP{1}$ as the outer norm, we obtain the Submodular Cost Allocation problem~\cite{chekuri-ene11:submodular-cost-allocation}. The authors obtain a $O(\log n)$-approximation.

Another line of research focuses on \emph{monotone} submodular functions as inner objective. In the Submodular Facility Location problem~\cite{svitkina-tardos10:hierarchical-facility-location} each center (facility) $x\in F$ is associated with a monotone, submodular function $s_x\colon P\rightarrow\nnr$. A solution $\sigma\colon P\rightarrow F$ assigns each point (client) to a center (facility). The goal is to minimize the total connection cost $\sum_{p\in P}\delta(p,\sigma(p))$ plus the facility cost $\sum_{x\in F}s_x(\sigma^{-1}(x))$. Similar to the other allocation problems, it does not concern selection of facilities but assumes they are fixed. Svitkina and Tardos~\cite{svitkina-tardos10:hierarchical-facility-location} and give an $O(\log n)$-approximation algorithm for it, which is asymptotically best possible as the problem generalizes set cover~\cite{shmoys-etal04:facility-location-service-costs}. In a recent work, Abbasi et al.~\cite{abbasi-etal24:submodular-fac-loc} design a $O(\log\log n)$-approximation for the uniform case where every facility is assigned the same submodular function.

\subsection{Motivating Example}\label{sec:motivating-example}
As already mentioned, \NCC{\textsf{Top}}{ \LP{1} } interpolates between two important clustering problems, namely $k$-Median ($\ell=n$) and Min-Sum of Radii ($\ell=1)$.
We believe that the intermediate problems defined are of independent interest.
We demonstrate that through an example.

\begin{figure}[!ht]
  \centering
  \subfloat[][Dataset]{\includegraphics[width=.45\textwidth]{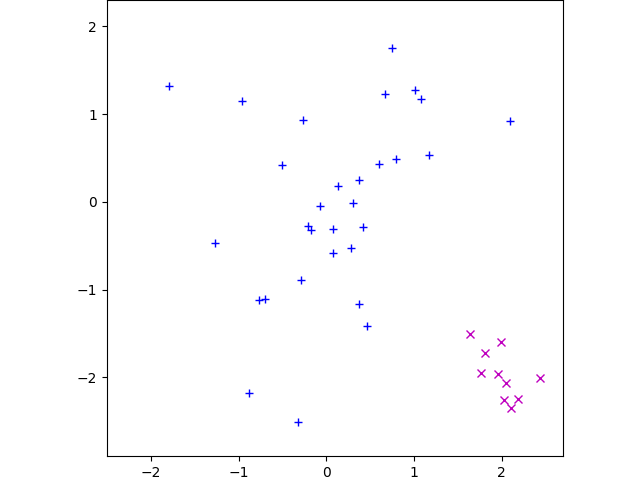}}\quad
  \subfloat[][$k$-Median]{\includegraphics[width=.45\textwidth]{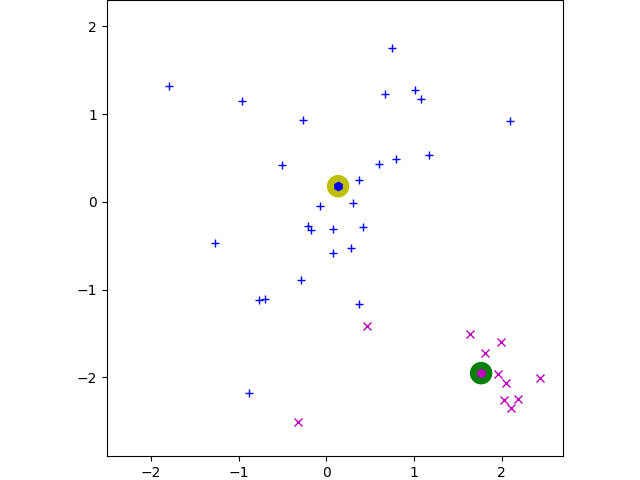}}\\
  \subfloat[][Min-Sum of Radii]{\includegraphics[width=.45\textwidth]{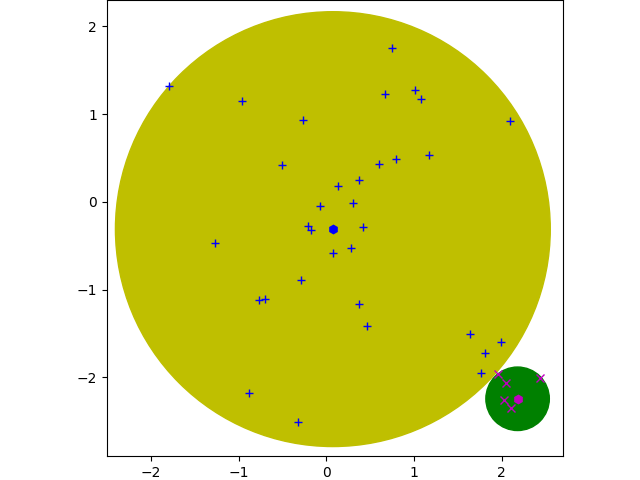}}\quad
  \subfloat[][\NCCS{\topl{8}{\cdot}}{\LP{1}}]{\includegraphics[width=.45\textwidth]{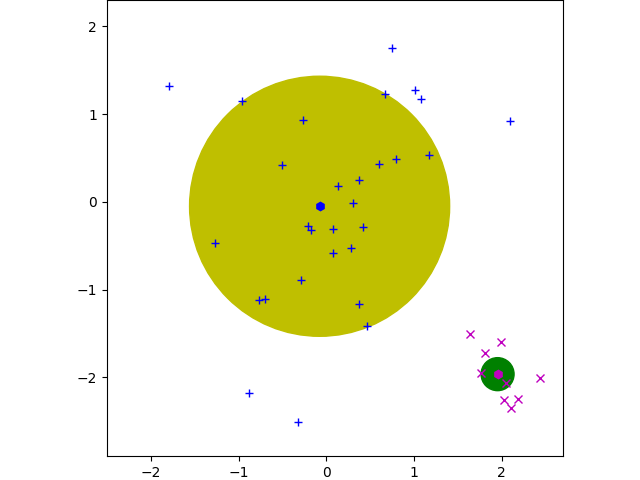}}
  \caption{A dataset of two clusters generated by two different $2$D Gaussians, on which we run $k$-Median, Min-Sum of Radii, and \NCCS{\topl{8}{\cdot}}{\LP{1}}. In the $k$-Median solution, points from the large cluster end up in the small cluster, while the opposite happens for Min-Sum of Radii. 
  Only \NCCS{\topl{8}{\cdot}}{\LP{1}} recovers the original clusters.
  In the fourth plot, the radii of the balls signify the $\ell$-th largest distance ($\ell=8$) in their respective clusters; these balls also directly relate to Ball $k$-Median (see \Cref{sec:topl}). 
  }
  \label{fig:dissection}
\end{figure}

In \Cref{fig:dissection} we have a dataset with two clusters generated by two different $2$D Gaussians.
We observe that the $k$-Median solution (and in fact any cluster-oblivious objective) assigns points of the large cluster to the small cluster.
This is a phenomenon that occurs with cluster-oblivious objectives; in fact, a motivation behind Min-Sum of Radii is that it reduces such dissection effects~\cite{hansen-jaumard97cluster-analysis,MinSumRadii}.
Nonetheless, the (inner) $\LP{\infty}$ norm in Min-Sum of Radii exhibits a high sensitivity towards variations in the periphery of a cluster, which may cause dissection effects as well. In our example, the ball capturing the large cluster has such a big radius that it cuts through the dense small cluster.

At the same time, if $\ell$ roughly corresponds to the number of points in the cluster peripheries then the (inner) $\topl{\ell}{\cdot}$ norm in \NCC{\textsf{Top}}{ \LP{1} } exhibits more robustness to variations in the peripheries and therefore can find a compromise between the two dissection issues.
In the example in \Cref{fig:dissection}, we were able to recover the original clusters.

\section{Preliminaries}
\begin{Definition}[\NNCC{}]\label{def:nncc}
    The input $\I = (P,F,\delta,k,f,g)$ consists of the point set~$P$, the set~$F$ of facilities, a metric~$\delta: (P\cup F) \times (P\cup F)\rightarrow \nnr$, a number~$k\in \mathbb{N}$, a symmetric, monotone norm~$f:\nnrvec\rightarrow \nnr$ and a symmetric, monotone norm~$g:\mathbb{R}^k_{\ge 0}\rightarrow \nnr$.
    A solution~$\X=(X,\sigma)$ consists of
    a subset~$X\subseteq F$ of facilities such that $|X|\le k$ and
    an assignment function~$\sigma\colon P \rightarrow X$. The goal is to find a solution~$\X$ that minimizes 
    \begin{align*}
        \cost{\sigma}{X} = g\left(\left( f(\distv{\sigma}{x} \right)_{x\in X}\right)
    \end{align*}
    where ${\distv{\sigma}{x} = (\dist{p}{x}\cdot\eins[x=\sigma(p)])_{p\in P} }$ is the cluster cost vector of $x$.
\end{Definition}

We define also special cases of \NNCC{}.

\begin{Definition}[\NCC{I}{O}]\label{def:ncc}
    For two classes $I$ and $O$ of norms this problem is the \NNCC{} with the restriction that $f\in I$ and $g\in O$.
\end{Definition}

For an instance $\I$ of \NNCC{}, let $\X^*_\I=(X^*_\I,\sigma^*_\I)$ be the optimal solution with value $\OPT_\I=\cost{\sigma^*_\I}{X^*_\I}$. 
For other problems that are defined later we adapt this notation.

\subsection{Different Types of Norms} \label{ssec:typesOfNorms}
Given a vector $\bm{x}$, we let $\bm{x}^\downarrow[i]$ be the $i$-th largest entry of $\bm{x}$.
We also define $\bm{x}^\downarrow$ to be the vector obtained by sorting the entries of $x$ in non-decreasing order.

We now introduce the different norms that are analyzed in this work.
The $\LP{1}$ norm is the sum of the absolute values of the entries of a vector and the $\LP{\infty}$ norm is the maximum absolute value of the entries of a vector.

We also define the top-$\ell$ norm.
\begin{Definition}[top-$\ell$ norm]
    For a number $\ell \in\mathbb{N}$ and a vector $\bm{x}\in \nnrvec$ with $n\ge \ell$ the $\topl{\ell}{\cdot}$-norm is defined as
    \begin{align*}
        \topl{\ell}{\bm{x}} = \sum_{i=1}^\ell x^\downarrow[i].
    \end{align*}
    Let $\textnormal{\textsf{Top}} =\{\topl{\ell}{\cdot}\mid \ell \in \mathbb{N}  \}$ be the class of all $\topl{\ell}{\cdot}$-norms.
\end{Definition}

Next, we define the ordered norm.
\begin{Definition}[ordered norm]
    For two vectors $\bm{x},\bm{w} = \{w_1,\dots,w_n\}\in \nnrvec$ where the entries of $\bm{w}$ are non-increasing, that is $w_i\ge w_{i+1}$ for all $i\in [n-1]$, the \ord{\bm{w}}{\cdot}-norm is defined as
    \begin{align*}
        \ord{\bm{w}}{\bm{x}} = \bm{w}\cdot\bm{x}^\downarrow = \sum_{i=1}^n w_i\cdot x^\downarrow[i].
    \end{align*}
    We call $\bm{w}$ the weight vector.
    Let $\textnormal{\textsf{Ord}} =\{\ord{\bm{w}}{\cdot}\mid \bm{w} \in \mathbb{R}^*_{\ge 0}  \}$ be the class of all $\ord{\bm{w}}{\cdot}$-norms.
\end{Definition}

\begin{Definition}
    Let $\textnormal{\textsf{Sym}}$ be the class of all symmetric, monotone norms.
\end{Definition}

\subsection{Proxy Costs}\label{sec:proxy}
The \topl{\ell}{\cdot} norm and the \ord{\bm{w}}{\cdot} are non-linear, which makes them difficult to work with.
To bypass this difficulty, \cite{chakrabarty-swamy19:norm-k-clustering} used proxy costs.
From a high level view, these proxy costs have an additional input (called the threshold); if we use the ``correct'' threshold, then the proxy cost (corresponding to some norm $f$) of a vector is equal to the $f$ norm of the vector.
Furthermore, no matter the choice of the threshold, the proxy cost of a vector upper bounds the $f$ norm of the vector.
The idea of proxy costs and the observations used in this section come from \cite{chakrabarty-swamy19:norm-k-clustering}.

We start with the proxy cost for the $\topl{\ell}{\cdot}$ norm. 
Let $\bm{x}= (x_1,\dots,x_n) \in \nnrvec$ be a vector and $y \in \nnr$.
We call $y$ the threshold.
\begin{align*}
    \proxy{y}{\bm{x}}{\ell} = \ell \cdot y  + \sum_{i=1}^n (x_i \dotdiv y)
\end{align*}
\begin{observation}
    For all vectors $\bm{x}\in \nnrvec$, thresholds $y\in \nnr$ and $\ell \in [n]$, it holds that
    \begin{align*} 
        \proxy{y}{\bm{x}}{\ell} \ge \topl{\ell}{\bm{x}}
    \end{align*}
\end{observation}

\begin{observation}
    For all vectors $\bm{x}\in \nnrvec$ and $\ell \in [n]$, it holds that
    \begin{align*}
        \proxy{\bm{x}^{\downarrow}[\ell]}{\bm{x}}{\ell} = \topl{\ell}{\bm{x}}
    \end{align*}
\end{observation}

Let $\bm{w}=(w_1,\dots,w_n)\in \nnrvec$, $\bm{x}= (x_1,\dots,x_n) \in \nnrvec$ and $\bm{t}=(t_1,\dots, t_n) \in \nnrvec$ be three vectors, where $\bm{w}$ and $\bm{t}$ are non-increasing.
We call $\bm{t}$ the threshold vector.

\begin{align*}
    \proxyz{\bm{x}}{\bm{w}}{\bm{t}} = \sum_{i=1}^n (w_i-w_{i+1})\proxy{t_i}{\bm{x}}{i}
\end{align*}

\begin{observation}
    For all vectors $\bm{w}=(w_1,\dots,w_n)\in \nnrvec$, $\bm{x}\in \nnrvec$, $\bm{t}\in \nnrvec$, it holds that
    \begin{align*} 
        \proxyz{\bm{x}}{\bm{w}}{\bm{t}} \ge \ord{\bm{w}}{\bm{x}}
    \end{align*}
\end{observation}

\begin{observation}
    For all vectors $\bm{w}=(w_1,\dots,w_n)\in \nnrvec,\bm{x}\in \nnrvec$, it holds that
    \begin{align*}
        \proxyz{\bm{x}}{\bm{w}}{\bm{x}^{\downarrow}} = \ord{\bm{w}}{\bm{x}}
    \end{align*}
\end{observation}

\section{\texorpdfstring{\NCC{\textsf{Top}}{\LP{1}}  and \Ballk{}}{\NCCH{\textsf{Top}}{\textnormal{l one}}  and Ball k-Median}} \label{sec:topl}
In this section, we design a constant factor approximation for \NCC{\textsf{Top}}{ \LP{1} }.

\apxtoplone*
The first step of our solution is to reduce \NCC{\textsf{Top}}{ \LP{1} } to a new problem called \Ballk{}, in \Cref{subsec:reduceballk}.
Then in \Cref{subsec:apxballk} we provide an approximation algorithm for \Ballk{}.

The input of \Ballk{} is a set of facilities and a set of clients in a metric space.
We can open $k$ balls around facilities and pay for their radius (times some parameter $\rho$).
In contrast to \msr{}, we do not need to cover all clients with these balls;
a client can also connect to a ball it is not covered by.
This incurs an additional cost, namely the distance of the client to the border of the ball. 
It is straightforward to verify that \Ballk{} is a generalization of both \msr{} ($\rho\le 1$) and \kmed{} ($\rho\ge n$).

On an intuitive level, we use \Ballk{} because it can be thought of as \NCCS{\textsf{Top}}{\LP{1}} with the difference that instead of \topl{\ell}{\bm{x}} we use \proxy{y}{\bm{x}}{\ell}. 
As argued in \Cref{sec:proxy} the proxy costs are easier to handle for approximation algorithms.
Usually, the threshold $y$ is guessed. 
In our case we have $k$ different thresholds, one for each facility. 
Thus, we cannot guess them efficiently.
We circumvent this problem by not having the threshold in the input, but in the solution to an instance.

\subsection{\texorpdfstring{Reduction to \Ballk{}}{Reduction to Ball k-Median}}\label{subsec:reduceballk}

In this section we define the \Ballk{} problem formally and reduce \NCCS{\textsf{Top}}{\LP{1}} to it. 
Note that \Cref{lem:redball} shows that the reduction also works in the other direction.
We start by introducing the new problem.

\begin{Definition}[\Ballk{}]\label{def:ball-k-median}
    The input $\I = (P,F,\delta,k,\rho)$ consists of the point set $P$, the set~$F$ of facilities, a metric~$\delta\colon (P\cup F) \times (P\cup F)\rightarrow \nnr$, a number~$k\in \mathbb{N}$ and the scaling factor~$\rho\in\mathbb{N}$.
    A solution~$\X = (X,r)$ contains a subset~$X\subseteq F$ of facilities such that $|X|\le k$ and a radius function~$r\colon X\rightarrow \nnr$. The goal is to find a solution $\X$ that minimizes
    \begin{align*}
        \costz{\X} = \sum_{p\in P} \distr{p}{X}{r} + \rho\sum_{x\in X}r(x),
    \end{align*}
    where $\distr{p}{X}{r} = \min_{x\in X}\distr{p}{x}{r}$ and $\distr{p}{x}{r} = \dist{p}{x}\dotdiv r(x)$.
\end{Definition}

We show that there is an approximation preserving reduction from \NCCS{\textsf{Top}}{\LP{1}} to \Ballk{} and vice versa. 

\begin{lemma}\label{lem:redball}
    Let $\I =(P,F,\delta,k,\topl{\ell}{\cdot},\LP{1})$ an instance of \NCCS{\textnormal{\textsf{Top}}}{\LP{1}}. Then the instance $\I'=(P,F,\delta,k,\rho =\ell)$ of \Ballk{} satisfies the following two properties.
    \begin{enumerate}
        \item For all solutions~$\X=(X,\sigma)$ for $\I$, we can compute a solution~$\X'=(X,r)$ for $\I'$ in polynomial time such that $\costz{\X'}\le \cost{\sigma}{X}$.
        \item For all solutions~$\X'=(X,r)$ for $\I'$, we can compute a solution~$\X=(X,\sigma)$ for $\I$ in polynomial time such that $\cost{\sigma}{X}\le \costz{\X'}$.
    \end{enumerate}
\end{lemma}
\begin{proof}
We show how to transform a solution $\X=(X,\sigma)$ for \NCCS{\textsf{Top}}{\LP{1}} to a solution $\X'=(X,r)$ of Ball $k$-median without increasing the cost. 
For every $x \in X$ set the radius $r(x)$ to the $\ell$-largest distance in the cluster cost vector $\distv{\sigma}{x}$. We denote this distance as $\ldist{\ell}{x}{\sigma}$.

\begin{align*}
    \costz{\X'}& = \sum_{p\in P} \distr{p}{X}{r}+ \sum_{x\in X}r(x)\rho\\
    &\le \sum_{p\in P}  \distr{p}{\sigma(p)}{r}+ \sum_{x\in X}r(x)\rho\\
    &= \sum_{x\in X}\sum_{p\in P: \sigma(p)=x}  \distr{p}{x}{r}+ \sum_{x\in X}r(x)\rho\\
    &= \left(\sum_{x\in X}\sum_{p\in P: \sigma(p)=x}  \dist{x}{p}\dotdiv \ldist{\ell}{x}{\sigma}\right)+ \sum_{x\in X}\ldist{\ell}{x}{\sigma}\ell\\
    &= \sum_{x\in X}\left(\left(\sum_{p\in P: \sigma(p)=x}  \dist{x}{p}\dotdiv \ldist{\ell}{x}{\sigma} \right)+\ldist{\ell}{x}{\sigma}\ell\right)\\
    &=\sum_{x\in X}\proxy{\ldist{\ell}{x}{\sigma}}{\distv{\sigma}{x}}{\ell}\\
    &=\sum_{x\in X}\topl{\ell}{\distv{\sigma}{x}}\\
    &=\cost{\sigma}{X}
\end{align*}

We also show how to transform a solution $\X'=(X,r)$ for Ball $k$-median to a solution $\X=(X,\sigma)$ for \NCCS{\textsf{Top}}{\LP{1}} without increasing the cost. 
For all points $p\in P$ define $\sigma(p) = \arg \min_{x\in X} \distr{p}{x}{r}$.
\begin{align*}
    \cost{\sigma}{X} & =\sum_{x\in X}\topl{\ell}{\distv{\sigma}{x}} \\&\le \sum_{x\in X}\proxy{r(x)}{\distv{\sigma}{x}}{\rho} \\ 
    &= \sum_{x\in X}\left(r(x)\rho+\sum_{p\in P: \sigma(p)=x}  \distr{p}{x}{r} \right)\\
    &= \sum_{x\in X}r(x)\rho + \sum_{p\in P}  \distr{p}{\sigma(p)}{r} \\
    &= \sum_{x\in X}r(x)\rho + \sum_{p\in P}  \distr{p}{X}{r}  \\
    &=\costz{\X'}
\end{align*}

\end{proof}

\subsection{\texorpdfstring{Overview of the Approximation Algorithm for \Ballk{}}{Overview of the Approximation Algorithm for Ball k-Median}}\label{subsec:apxballk}
We present a primal-dual algorithm that needs a combination of ideas from multiple algorithms for \kmed{}(\cite{JainVaz}\cite{ola}) and \msr{}(\cite{MinSumRadii}\cite{ahmadian}). 

In \Cref{subsubsec:guess} we guess a constant number of balls from the optimal solution.
This is a technical necessity to bound a special case in a later step, similar to \cite{MinSumRadii}.

In \Cref{subsubsec:lmp} we formalize the facility location version of \Ballk{} and show an LMP approximation algorithm for it. 
This is an approximation for the facility location version where we are more strict towards the opening costs.
The algorithm combines ideas from \cite{JainVaz} and \cite{MinSumRadii}. 
The dual-ascent phase is an adaption from \cite{JainVaz} and the pruning phase uses the greedy approach from \cite{MinSumRadii}.

In \Cref{subsubsec:binary} we use the standard binary search technique to obtain a good bi-point solution using the LMP approximation.

In \Cref{subsubsec:bipointrounding} we obtain a solution by rounding the bi-point solution. 
We use a generalization of grouping techniques from \cite{ahmadian} and \cite{ola} which allows us to decide which facilities to open by solving a knapsack-LP.

We note that obtaining a bi-point solution from an LMP approximation to the facility location version of the problem, and using it to obtain an approximation for the original problem is a standard framework that has found success in several clustering problems.
Our novelty is in the design of the LMP approximation and in the rounding of the bi-point solution.

In particular, regarding the rounding of the bi-point solution, there exists one ``special'' group after solving the knapsack-LP which was previously handled in two different fashions.
In the context of \msr{} it was sufficient to only open one ball from this group\cite{ahmadian}.
In the context of \kmed{} \cite{ola} opened more facilities from this group than the cardinality constraint would allow. 
However, they show that violating this constraint by an additive constant can be ``fixed''.

In our case, we use a novel approach, namely we open slightly less balls in the special group than indicated by the LP. 
We proceed to show that this at most triples the number of balls that are not opened in the special group.
This insight is sufficient to bound the cost of the solution we create.

\subsection{Guessing the Largest Balls}\label{subsubsec:guess}
The idea of this section originates from \cite{MinSumRadii}. 
As in the final step of our algorithm, namely the bi-point rounding in \Cref{subsubsec:bipointrounding}, we are able to bound the cost of all but one ball, we shall ensure that the size of this particular ball is small.
We achieve that by first guessing the $O(\nicefrac{1}{\eps})$ largest balls of the optimal solution.

In contrast to \cite{MinSumRadii}, we cannot remove these guessed balls (and the clients covered by them) and simply work on the reduced instance.
That is because, by our problem's definition, clients not inside a ball can still connect to it, by paying an amount equal to their distance to the perimeter of it.

Algorithm~\ref{alg:guessopt} guesses the $O(\nicefrac{1}{\eps})$ facilities (and their respective radii) with the largest radii, from the optimal solution {$\X^*_\I=(X^*_\I,r^*_\I)$} of instance $\I$. This ensures that the remaining steps of the algorithm only have to deal with balls of cost at most $O(\eps \OPT_{\I'})$.
We assume without loss of generality that for all $x\in X^*_\I$, we have $r(x)=\dist{p}{x}$ for some $p\in P$ by applying~\Cref{lem:redball} twice.
\begin{algorithm2e}[!ht]
  \SetKwFunction{Guess}{GUESS}

  \setcounter{AlgoLine}{0}
  \SetKwProg{procedure}{Procedure}{}{}
  
   \procedure{\Guess{$\I = (P,F,\delta,k,\rho), \eps$}}{
  $t\gets \ceil{\nicefrac{3}{\eps}}$\;
  \ForEach{$T\in \binom{F}{t}$ and $r\colon T\rightarrow \{\dist{p}{x}\mid p\in P, x\in T\} $}{
    \textbf{output} $(T,r)$\;
    }
  }

\caption{Guessing the largest balls.}
\label{alg:guessopt}

\end{algorithm2e}

We prove that at least one of the pairs that Algorithm~\ref{alg:guessopt} outputs guesses the largest balls of the optimal solution correctly. 

\begin{lemma}\label{lem:guessopt}
    Algorithm~\ref{alg:guessopt} outputs in time $n^{O(\nicefrac{1}{\eps})}$ a list of pairs so that at least one pair, say $(T,r)$, satisfies
    \begin{itemize}
        \item $T\subseteq X^*_{\I}$ and for all $x\in T$ it holds that $r(x)=r^*_\I(x)$,
        \item For all $x\in X^*_\I\setminus T$ it holds that $r^*_\I(x) \le  \min_{x'\in T} r(x')$
        \item $\min_{x'\in T}r(x') \le \eps / (3\rho)\cdot\OPT_\I$. 
    \end{itemize}
\end{lemma}
\begin{proof}
    Our algorithm iterates over all $t$-element subsets of $F$ along with all radii functions that map to facility-client distances. Therefore, in one iteration the algorithm outputs the set $T$ of the $t$ facilities in $\X^*_\I$ with the largest radii and the correct corresponding radii function $r$. 
    Since these are the largest radii, it follows for all $x\in X^*_{\I}\setminus T$ that $r^*_\I(x) \le \min_{x'\in T}r(x')$. 
    We use an averaging argument to bound the smallest guessed radius
    \begin{align*} 
        \rho\cdot \min_{x\in T}r(x) &\le \rho \cdot \frac{\sum_{x\in T}r(x)}{|T|}\\
        &\le  \frac{\rho\cdot \sum_{x \in X^*_\I}r^*_\I(x)}{\ceil{\frac{3}{\eps}}}\\
        &\le \frac{\eps}{3}\OPT_\I\,.
    \end{align*}
\end{proof}

\subsection{LMP Approximation Algorithm}\label{subsubsec:lmp}
In this section we introduce the facility location version of \Ballk{} and show a Lagrange multiplier preserving (LMP) approximation algorithm for it.

\subsubsection{\FLBall{}}
The facility location version of~\Ballk{} is its Lagrange relaxation w.r.t.\ the solution size. Specifically, we remove the cardinality constraint $|X|\le k$ but penalize the solution size $|X|$ in the objective function by charging a fixed cost $\lambda$ for every open facility (on top of their radius-dependent cost).
The formal definition is as follows.

\begin{Definition}[\FLBall{}]
    The input $\I = (P,F,\delta,\rho,\lambda)$ consists of the point set $P$, the set~$F$ of facilities, a metric~$\delta\colon (P\cup F) \times (P\cup F)\rightarrow \nnr$, the scaling factor~$\rho\in\mathbb{N}$ and an opening cost $\lambda \ge 0$.
    A solution~$\X = (X,r)$ contains a subset~$X\subseteq F$ of facilities and a radius function~$r\colon X\rightarrow \nnr$. The goal is to find a solution $\X$ that minimizes\
    \begin{align*}
        \costz{\X} +|X|\lambda= \sum_{p\in P} \distr{p}{X}{r} + \sum_{x\in X}(r(x)\rho + \lambda).
    \end{align*}
\end{Definition}

\subsubsection{LP Relaxation for \FLBall{}} 
    
    In Figure~\ref{fig:FLLP}, we introduce our LP relaxation for \FLBall{} where a part $(T,\hat{r})$ of the solution (the guessed part) is fixed.
    Let $R_x$ be the set of allowed radii after the guessing, that is $R_x = \{\dist{x}{p}\mid p \in P \textnormal{ and } \dist{x}{p}\le \min_{x'\in T}\hat{r}(x')\}$ for $x\in F\setminus T$ and $R_x=\{\hat{r}(x)\}$ for $x\in T$.
    
    Variable~$v_{xp}^r$ for $x\in F, p\in P$ and $r\in R_x$ indicates whether client $p$ is connected to facility $x$ with a ball of radius $r$. Variable~$u_x^r$ indicates whether facility $x$ is opened with a ball of radius $r$.
    The first constraint demands that each client is connected to at least one facility.
    The second constraint demands that a client can only be connected to an opened ball.  Note that we extend the notation $\distr{\cdot}{\cdot}{r}$ to the case where $r$ is a number rather than a function.
For a client $p\in P$, a facility $x\in F$ and a radius $r\in \nnr$, let $\distr{p}{x}{r}=\dist{p}{x}\dotdiv r$.
    
    Note that this LP is similar to the standard LP for the Uniform Facility Location Problem if we understand each ball (pair of facility $x\in F$ and radius $r\in R_x$) as an independent facility. 
    There are two main differences.
    Firstly, the cost $\distr{\cdot}{\cdot}{r}$ of connecting a client to a ball is not a metric because we substract the radius from the distance to the facility.
    Secondly, the opening cost of a ball consists not only of the fixed cost $\lambda$ but also a radius-dependent cost $\rho r$. 
\begin{figure}[!ht]
\caption{\,LP for facility-location Ball $k$-median and predetermined $(T,\hat{r})$.} \label{fig:FLLP}
\vskip -1.5ex\rule{\linewidth}{.5pt}
\begin{mini}|s|<b>{}{\sum_{x\in F, p\in P,r\in R_x} \distr{p}{x}{r}v_{xp}^r +  \sum_{x\in F\setminus T,r\in R_x}(\rho r+\lambda)u_x^r}{}{}
    \addConstraint{\sum_{x\in F ,r\in R_x}v_{xp}^r }{\geq 1\qquad}{\forall p \in P}
    \addConstraint{v_{xp}^r}{\le u_x^r}{\forall x \in F, p\in P, r\in R_x}
\end{mini}
\rule{\linewidth}{.5pt}
\end{figure}

\subsubsection{LMP Approximation}
Based on the approach by Jain and Vazirani~\cite{JainVaz} and using ideas from Charikar and Panigrahi~\cite{MinSumRadii} we can design an LMP factor-3 approximation algorithm for \FLBall{} by exploiting properties of the distance measure $\distr{\cdot}{\cdot}{r}$. However, here we prove the following more technical lemma that relates the cost of the facility location instance to the cost of the optimum solution of the underlying \Ballk{} instance. This is required because we need to incorporate the guessed balls in set $T$ into our solution, which is not necessary in the algorithm from~\cite{MinSumRadii} for~\msr{}.

Intuitively, an LMP approximation is an approximation that is stricter towards the opening costs.
Concretely, for \FLBall{} the following lemma formalizes the factor-$3$ LMP approximation.
   \begin{restatable}{lemma}{lemlmpalgo}
\label{lem:lmpalgo}
    
    Let $\I = (P,F,\delta,k,\rho)$ be an instance of \Ballk{}, $\lambda\ge 0$, $\eps > 0$ and $(T,\hat{r})$ be such that the properties of Lemma~\ref{lem:guessopt} are satisfied. 
    Then Algorithm~\ref{alg:lmpflball} efficiently computes a pair $\X=(X,r)$ such that \begin{align*}
        \costz{\X} + 3\lambda|X| \le 3(\OPT_{\I}+\lambda k),
    \end{align*} 
     $T\subseteq X$, $r(x) = r^*_\I(x)$ for all $x\in T$, and $\max_{x\in X\setminus T} r(x)\le \frac{\eps}{\rho} \OPT_{\I}$.
\end{restatable}
   In the remainder of this section, we prove this lemma.

Our LMP approximation algorithm follows the primal-dual framework; see Figure~\ref{fig:dual} for the dual of our LP relaxation.
\begin{figure}[!ht]
\caption{\,Dual-LP for Figure~\ref{fig:FLLP}.} \label{fig:dual}
\vskip -1.5ex\rule{\linewidth}{.5pt}
\begin{maxi}|s|<b>{}{\sum_{p\in P} \alpha_p}{}{}
    \addConstraint{\alpha_p-\beta_{xp}^r}{\le \distr{p}{x}{r}\qquad}{\forall x \in F, p\in P,r\in R_x}
    \addConstraint{\sum_{p\in P}\beta_{xp}^r}{\leq \lambda +\rho r}{\forall x \in F\setminus T, r\in R_x}
    \addConstraint{\sum_{p\in P}\beta_{xp}^r}{\leq 0}{\forall x \in T,x\in R_x}
\end{maxi}
\rule{\linewidth}{.5pt}
\end{figure}

On a high level, the algorithm consists of two phases.
The first phase is the dual-ascent phase(Lines~\ref{al:firstline} to~\ref{al:endgrowing}). 
Here, the algorithm increases the dual variables until constraints(cf.\ Lines~\ref{al:feasiblea}, \ref{al:feasibleb} and \ref{al:feasiblec} in Algorithm~\ref{alg:lmpflball}) in the Dual-LP in Figure~\ref{fig:dual} get tight. 
From these tight constraints the algorithm produces a set of candidate balls.
In the pruning phase(Lines~\ref{al:beginopening} to~\ref{al:endopening}) the algorithm greedily selects an ``independent'' subset of the candidate balls as the solution.

The algorithm is described more detailed in Algorithm~\ref{alg:lmpflball}.
Additionally to the integral primal solution, it produces a (fractional) dual solution.
The costs of the primal solution can be bounded using a constant multiple of the dual solution. 
Thus, we can guarantee the quality of the primal solution.

Note that Lines~\ref{al:dualinita} and~\ref{al:dualinitb} initialize the dual variables to zero which is a valid solution. 
Furthermore, Lines~\ref{al:feasiblea},~\ref{al:feasibleb} and~\ref{al:feasiblec} maintain the validity of the dual solution.
\begin{observation}
    Upon termination of Algorithm~\ref{alg:lmpflball}, the variables $\alpha_p$ for $p\in P$ and $\beta_{xp}^r $ for ${p\in P,x\in F,r\in R_x}$ form a valid solution to the dual LP in \Cref{fig:dual}.
\end{observation} 
The dual-ascent phase of the algorithm (Lines~\ref{al:firstline} to~\ref{al:endgrowing}) is an adaptation of \cite{JainVaz}.
The novelty is that we adapted to the differences of the LP in \Cref{fig:FLLP}, compared to the standard LP for the Facility Location Problem.

The pruning phase of the algorithm (Lines~\ref{al:beginopening} to~\ref{al:endopening}) needs to combine ideas from different algorithms.
We need to use the notation of conflict from \cite{JainVaz}, namely that two facilities are in conflict when a client contributed to both (cf.\ Line~\ref{al:conflict}).
However, in our case it is not sufficient to just pick all the facilities, because of technical reasons related to $\distr{\cdot}{\cdot}{r}$ not being a metric. 
We also need to triple the radii of their balls as in \cite{MinSumRadii} (cf.\ Line~\ref{al:triple}). 

\begin{algorithm2e}[!ht]
  \SetKwFunction{Bal}{\Ballk{}}

  \setcounter{AlgoLine}{0}
  \SetKwProg{procedure}{Procedure}{}{}
  
   \procedure{\Bal{$\I = (P,F,\delta,\rho,\lambda), (T,\hat{r})$}}{
  $Y\gets T$\label{al:firstline}\;
  $r'(x)\gets 0$ for all $x\in F\setminus T$\;
  $\alpha_p \gets 0$ for all $p\in P$\label{al:dualinita}\;
  $\beta_{xp}^r \gets 0$ for all $p\in P,x\in F,r\in P_x$\label{al:dualinitb}\;
  Start increasing all $\alpha_p$ simultaneously at the same rate\;
  \While{$\alpha_p$ is increasing for some $p\in P$}{
    \If{$\alpha_p-\beta_{xp}^r =\distr{p}{x}{r}$ for some $p\in P, x\in F,r\in R_x$\label{al:feasiblea}}{
        Start increasing $\beta_{xp}^r$ at the same rate as the $\alpha_p$\;
    }
    \If{$\sum_{p\in P}\beta_{xp}^r = \rho r+\lambda$ for some $x\in F\setminus T,r\in R_x$\label{al:feasibleb}}{
        $Y\gets Y\cup \{x\}$\label{al:addtoy}\;
        $r'(x)\gets \max\{r'(x),r\}$\label{al:maxra}\;
        Stop increasing $\beta_{xp}^r$\;
        Stop increasing $\alpha_p$ for all $p$ such that $\alpha_p-\beta_{xp}^r = \distr{p}{x}{r}$\;
    }
    \If{$\sum_{p\in P}\beta_{xp}^r = 0$ for some $x\in  T,r\in R_x$\label{al:feasiblec}}{
        $Y\gets Y\cup \{x\}$\;
        $r'(x)\gets \max\{r'(x),r\}$\label{al:maxrb}\;
        Stop increasing $\beta_{xp}^r$\;
        Stop increasing $\alpha_p$ for all $p$ such that $\alpha_p-\beta_{xp}^r = \distr{p}{x}{r}$\label{al:endgrowing}\;
    }
    }
  
  $r(x)\gets \hat{r}(x)$ for all $x\in T$\;
  $X\gets T$\label{al:beginopening}\;
  $Y\gets Y\setminus T$\;
  \While{$Y\ne \emptyset$}{
  Pick $x = \arg\max_{x'\in Y}r'(x')$\;
  $X\gets X\cup \{x\}$\;
  $r(x)\gets 3r'(x)$\label{al:triple}\;
  $Y\gets Y \setminus \left\{x'\in Y\mid \beta_{xp}^{r'(x)},\beta_{x'p}^{r'(x')}>0\text{ for some } p\in P\right\}$\label{al:conflict}\;
  }
  \Return $(X,r)$\label{al:endopening}\;
}
\caption{Approximate \FLBall{}.}
\label{alg:lmpflball}

\end{algorithm2e}

\subsubsection{Analysis}
In the analysis of Algorithm~\ref{alg:lmpmsr} we distinguish two types of clients: Clients that contribute towards opening a facility (ball) and those that do not contribute. We argue that the dual variables of contributing clients can fully pay for their connection costs, the fixed costs of the balls opened, and a third of their radius-dependent costs. For non-contributing clients, we argue that their dual budget can pay at least a third of their connection cost. This distinction between contributing and non-contributing clients is also done in \cite{JainVaz}, but not in \cite{MinSumRadii} because in \msr{} there are only opening costs and no connection costs.

 More formally, the set $P_x = \{ p\in P \mid \beta_{xp}^{r'(x)}>0 \}$ denotes the set of clients that contribute towards opening facility $x\in X$ with radius $r'(x)$. We denote by $P_{X} = \bigcup_{x\in X}P_x$ the set of all contributing clients. 

\begin{lemma}[Contributing Clients] \label{claim:alphaballs}
        Let $\I = (P,F,\delta,k,\rho)$ be an instance of \Ballk{}, $\lambda\ge 0$, $\eps > 0$ and $(T,\hat{r})$ be such that the properties of Lemma~\ref{lem:guessopt} are satisfied.
        Upon termination of Algorithm~\ref{alg:lmpflball} we have that
        \begin{align*}
            \sum_{x\in X\setminus T} \sum_{p\in P_x}\distr{p}{x}{r'} + \rho\sum_{x\in X\setminus T} r'(x) + |X\setminus T|\lambda = \sum_{p\in P_X}\alpha_p.
        \end{align*}\
    \end{lemma}
    \begin{proof}
        By the design of the Algorithm~\ref{alg:lmpflball}, we know that for a facility $x\notin T$ that has been added to $Y$ in Line~\ref{al:addtoy} it holds that
        \begin{align*}
            \sum_{p\in P_x}\beta_{xp}^{r'(x)} = \lambda+\rho r'(x).
        \end{align*}
        Additionally, we know for a client $p\in P_x$ that 
        \begin{align*}
            \alpha_p-\beta_{xp}^{r'(x)} = \distr{p}{x}{r'}.
        \end{align*}
        
        By the choice of $X$ we know that for any two facilities $x_1,x_2\in X$ there is no client contributing to opening both, that is there is no client $p\in P$ such that $\beta_{x_1p}^{r'(x_1)}>0$ and $\beta_{x_2p}^{r'(x_2)}>0$. This implies that $P_{x_1}$ and $P_{x_2}$ are disjoint. Additionally, note that $P_{\overline{x}}=\emptyset$ for all $\overline{x}\in T$ because by Line \ref{al:feasiblec} all the corresponding dual variables satisfy $\beta_{\overline{x}\overline{p}}^{\overline{r}}=0$ for all $\overline{p} \in P$ and $\overline{r} \in R_{\overline{x}}$.
        Thus, we get.
        \begin{align*}
            \sum_{p \in P_X} \alpha_p &=\sum_{x\in X} \sum_{p\in P_x}\alpha_p\\
            &= \sum_{x\in X\setminus T}\left (  \sum_{p\in P_x}(\alpha_p-\beta_{xp}^{r'(x)})+\sum_{p\in P_x}\beta_{xp}^{r'(x)} \right)\\
            &= \sum_{x\in X\setminus T} \left(\sum_{p\in P_x}\distr{p}{x}{r'} + \rho r'(x) + \lambda \right)\\
            &=\sum_{x\in X\setminus T}\ \sum_{p\in P_x}\distr{p}{x}{r'}+\sum_{x\in X\setminus T}(\rho r'(x) + \lambda) 
        \end{align*}
    \end{proof}

    We upper bound the connection cost of clients $p\in P\setminus P_X$ that do not contribute to opening a facility by thrice the value of their dual variable $\alpha_p$.
    \begin{lemma}[Non-contributing Clients] \label{claim:alphadirect}
        Let $\I = (P,F,\delta,k,\rho)$ be an instance of Ball $k$-median, $\lambda\ge 0$, $\eps > 0$ and $(T,r')$ be such that the properties of Lemma~\ref{lem:guessopt} are satisfied.
        Upon termination of Algorithm~\ref{alg:lmpflball} we have that for all clients $p\in P \setminus P_X$ there is a facility $x\in X$ such that
        \begin{align*}
            \distr{p}{x}{r} \le 3\alpha_p.
        \end{align*}
    \end{lemma}
    \begin{proof}
        Let $p$ be an arbitrary element in $P\setminus P_X$ and $x'$ be the facility and $\Tilde{r}$ be the radius that caused $\alpha_p$ to stop increasing. If $x' \in X$, then we know
        \begin{align*}
               \distr{p}{x'}{r} \le\distr{p}{x'}{r'} \le\distr{p}{x'}{\Tilde{r}}=  \alpha_p-\beta_{x'p}^{\Tilde{r}} \le \alpha_p \le 3\alpha_p
        \end{align*}
        because by Lines~\ref{al:maxra} and~\ref{al:maxrb} it holds that $\Tilde{r}\le r'(x')$.
        Thus, we can assume from now on that $x' \notin X$. 

        Let $x \in X$ and $p'\in P$ such that $p$ contributed to both $x$ and $x'$, that is $\beta_{xp'}^r,\beta_{x'p'}^r>0$, and $r'(x')\le r'(x)$. We know such an $x$ exists by the greedy choice of $X$.
        Note that $x\notin T$ because by Line~\ref{al:feasiblec} for all facilities $\overline{x}\in T$ all the corresponding dual variables satisfy $\beta_{\overline{x}\overline{p}}^{\overline{r}}=0$ for all $\overline{p}\in P$ and $\overline{r}\in R_{\overline{x}}$.
        
        Generally, we observe.
        \begin{align*}
             \distr{p}{x}{r} = &\distr{p}{x}{3r'}
             \\=&\dist{x}{p}\dotdiv 3 r'(x) \\
             \le& (\dist{x}{p'}+\dist{x'}{p'}+\dist{x'}{p})\dotdiv 3r'(x) \\
             \le &(\dist{x}{p'}\dotdiv r'(x))+(\dist{x'}{p'}\dotdiv r'(x))+(\dist{x'}{p}\dotdiv r'(x))\\
             \le &(\dist{x}{p'}\dotdiv r'(x))+(\dist{x'}{p'}\dotdiv r'(x'))+(\dist{x'}{p}\dotdiv r'(x'))\\
             =&\distr{p'}{x}{r'}+\distr{p'}{x'}{r'}+\distr{p}{x'}{r'}\\
             \le &\alpha_p + 2\alpha_{p'}
        \end{align*}
        Since $p'$ is responsible for the edge between $x'$ and $x$, we know $\beta_{x'p'} > 0$ and $\beta_{xp'}>0$. The facility $x'$ was responsible for $\alpha_p$ to stop increasing. Thus, we observe $t_{x'} \le \alpha_p$ where $t_{x'}$ is the point in time where $x'$ was opened. Because $p'$ contributed to opening $x'$ ($\beta_{x'p'} > 0$), we know that $\alpha_{p'}$ did not increase after $x'$ was opened ($\alpha_{p'} \le t_{x'}$). Hence, we observe $\alpha_p\ge \alpha_{p'}$. This concludes the proof of the claim.
    \end{proof}

Since we bounded all the cost by (multiples of) dual variables, we can prove the approximation guarantee of Algorithm~\ref{alg:lmpflball}.
This concludes the main lemma of this section.

\lemlmpalgo*
\begin{proof}
We start by showing that all radii in $X\setminus T$ are small, that is $r(x)\le \nicefrac{\eps}{\rho}\cdot\OPT_{\I}$ for all $x\in X\setminus T$. 
    By the definition of $X$ and the design of the LP relaxation in~\Cref{fig:FLLP}, we know for all $x\in X\setminus T$, that $\rho  r'(x)\le \rho\min_{x\in T}\hat{r}(x)\le \nicefrac{\eps}{3}\OPT_\I$. 
    Since the radii are tripled in the construction of $r$, we conclude $\rho r(x)\le \eps \OPT_\I$.
    
    In the remaining part of the proof, we focus on bounding the cost of $(X,r)$.

    \begin{align}
        &\costz{X,r} + 3\lambda |X\setminus T|- \sum_{x\in T} r(x)\rho\nonumber\\
        \le &\sum_{p\in P} \distr{p}{X}{r} + \sum_{x\in X\setminus T}r(x)\rho +3\lambda|X|\nonumber\\
        \le& \sum_{x\in X}\sum_{p\in P_x}  \distr{p}{x}{r}  + \sum_{x\in X\setminus T}r(x)\rho +3\lambda|X\setminus T|+\sum_{p\in P\setminus P_X}\distr{p}{X}{r}\label{Line:split}\\
        =& \sum_{x\in X}\sum_{p\in P_x} \distr{p}{x}{r} + \sum_{x\in X\setminus T}3r'(x)\rho +3\lambda|X\setminus T|+\sum_{p\in P\setminus P_X}\distr{p}{X}{r}\nonumber\\
        \le&~3\sum_{p\in P_X}\alpha_p +\sum_{p\in P\setminus P_X}\min_{x\in X} \distr{x}{p}{r}\label{Line:alphasf}\\
        \le&~3\sum_{p\in P} \alpha_p\label{Line:alphas}\\
        \le&~3 \sum_{p\in P} \distr{p}{X^*_\I}{r^*_\I}+ 3\sum_{x\in X^*_\I \setminus T}\left(r^*_\I(x)\rho +\lambda\right)\label{Line:InsertOpt}
    \end{align}
    In (\ref{Line:split}) we split the sum over $P$ into two sums over $P_X$ and $P\setminus P_X$ and for $P_X$ we do not take the closest ball but the ball it contributed to. 
    (\ref{Line:alphasf}) is a consequence of \Cref{claim:alphaballs}  and (\ref{Line:alphas}) is a consequence of \Cref{claim:alphadirect}. 
    (\ref{Line:InsertOpt}) is true because $X^*_\I$ is a valid solution to the LP and the value of a dual solution is always at most the value of any primal solution.
    
    By adding $3\sum_{x\in T}\left(r(x)\rho + \lambda\right)$ to the terms above, this proof is completed.
\end{proof}

\subsection{\texorpdfstring{Binary Search on $\lambda$}{Binary Search on lambda}}\label{subsubsec:binary}
We can use \Cref{lem:lmpalgo} to obtain a bi-point solution; we prove this in \Cref{lem:bipoint}.
Let us note that the proof directly follows from standard techniques.
Therefore, we present the proof in the appendix.

\begin{Definition}
Let $\X_1=(X_1,r_1),\X_2=(X_2,r_2)$ be two solutions with $|X_1| \le k < |X_2|$ (in particular, for $a=(|X_2|-k)/(|X_2|-|X_1|)\ge 0$ and $b =(k-|X_1|)/(|X_2|-|X_1|)\ge 0$
we have $a+b=1$ and $a|X_1| + b|X_2| = k$).

Then we say that $(\X_1,\X_2)$ is a bi-point solution.
\end{Definition}

\begin{lemma} \label{lem:bipoint}
Given an instance $\I=(P,F,\delta,k,\rho)$, $\eps>0$, and a pair $(T,\hat{r})$, that satisfies \Cref{lem:guessopt}, in polynomial time we can obtain a bi-point solution $\X_1=(X_1,r_1),\X_2=(X_2,r_2)$ such that:
\begin{itemize}
    \item $|X_1| \le k < |X_2|$ (in particular, we can compute $a,b \ge 0$ such that $a+b=1$ and $a|X_1| + b|X_2| = k$).
    \item $a \costz{\X_1} + b \costz{\X_2} \le (3+\eps) \OPT_\I$.
    \item $T\subseteq X_1$ and $T \subseteq X_2$.
    \item $r_1(x) = r^*_\I(x) = r_2(x)$ for all $x\in T$.
    \item If $x_1 \in X_1 \setminus T$ and $x_2\in X_2\setminus T$, then $ r_1(x_1) < \nicefrac{\eps}{\rho} \OPT_\I$ and $r_2(x_2) < \nicefrac{\eps}{\rho} \OPT_\I$.
\end{itemize}
\end{lemma}

\subsection{Bi-Point Rounding}
\label{subsubsec:bipointrounding}
In this section we show how to design a constant-factor approximation for \Ballk{}.
Following a standard paradigm (see e.g.\ \cite{ahmadian,ola}), our starting point is a bi-point solution.

In \Cref{subsubsec:binary} we prove that for any $\eps>0$ we can obtain a bi-point solution $(\X_1,\X_2)$ for which $a\cdot \costz{\X_1}+b\cdot \costz{\X_2} \le (3 + \eps)\OPT_{\I}$.
This bi-point solution also has the technical property that any balls that are not in the guess $(T,\hat{r})$  have small radius.
We show how to obtain a solution $\X=(X,r)$ such that $|X|\le k$ and $\costz{\X} \le (13.5+7.5\eps)\OPT_\I$.
In particular, in the rest of this section we prove the following lemma:

\begin{restatable}{lemma}{ballkmed} \label{thm:ballkmed}
    Suppose we are given an instance of \Ballk{} $\I =(P,F,\delta,k,\rho)$, $\eps >0$, a pair $(T,\hat{r})$ and a a bi-point solution $(\X_1,\X_2)$ such that $a\cdot \costz{\X_1}+b\cdot \costz{\X_2} \le (3 + \eps)\OPT_{\I}$, $T\subseteq X_1,X_2$, for all $x\in T$ it holds that $\hat{r}(x) = r^*_\I(x)$ and for every $x_1\in X_1\setminus T, x_2\in X_2\setminus T$ we have $r_1(x_1),r_2(x_2) \le \nicefrac{\eps}{\rho} \OPT_\I$.
    We can design a solution $\X=(X,r)$ with $|X|\le k$ and ${\costz{\X}\le  (13.5+7.5\eps)\OPT_\I}$.
\end{restatable}

Consider $\I, \eps, \X_1, \X_2, a, b, T ,\hat{r}$ fixed throughout this section.
Furthermore, as stated in \Cref{thm:ballkmed}, for $(X_1,X_2)$ we assume that if $x_1\in X_1\setminus T$ and $x_2 \in X_2\setminus T$, then $r_1(x_1) \le \nicefrac{\eps}{\rho} \OPT_\I$ and $r_2(x_2)\le \nicefrac{\eps}{\rho} \OPT_\I$.

$\X_1$ is itself a $(3+\eps)/a$ approximation.
If $a > \nicefrac{1}{4}$, we thus obtain a $4(3+\eps)$ approximation using at most $k$ facilities.
Similarly, if $\costz{\X_1} \le \costz{\X_2}$ then $\X_1$ is a $(3+\eps)$ approximation because $\costz{\X_1} = (a+b)\costz{\X_1} \le a\costz{\X_1} + b\costz{\X_2} \le (3+\eps)\OPT_\I$.

From this point on we assume $a\le \nicefrac{1}{4}$ and $\costz{\X_1} > \costz{\X_2}$.
Notice that this means $\costz{\X_2} \le a\cdot \costz{\X_1} + b\cdot \costz{\X_2} \le (3+\eps) \OPT_\I$.

From a high level view, we group every facility $x_2$ from $X_2$ to its ``closest'' facility $cl_1(x_2)$ in $\X_1$.
Then for every group we either decide to open all facilities from $\X_2$, or the single facility from $\X_1$.
We use linear programming to decide what to do for each group.
We also use properties of the linear program to analyze the cost of our solution.
For technical reasons we have an exception, one group where we open the single facility from $X_1$ and some facilities from $X_2$.
Here, we need the fact that the radii of balls that are not in the guess $(T,\hat{r})$ are small.

Each client $p$ is assigned to the facility $x$ that minimizes $\dist{x}{p}-r(x)$ (we break ties arbitrarily).
The main challenge here is arguing about clients whose closest facility $x_2$ from $\X_2$ is not open in $\X$.
In these cases, we always have $cl_1(x_2)$ open in $\X$.
In \cite{ola} the authors used triangle inequality to bound the cost to connect the client to $cl_1(x_2)$.
Since distance costs in our case are not metric, we need to increase its radius by an amount related to the radii of the facilities of $\X_2$ in its group and use a generalized version of the triangle inequality.

\subsubsection{Grouping of Balls}
Let us now be more formal:
We expand the domain of $r_1$ to $X_1 \cup P$ by defining $r_1(p)=0$ for all $p\in P$ and the domain of $r_2$ to $X_2\cup P$ by defining $r_2(p)=0$. Given $x_1 \in X_1 \cup P$ and $x_2\in X_2 \cup P$, we define the distance between the balls around $x_1$ and $x_2$ as $\dists{x_1}{x_2} = (\dist{x_1}{x_2}\dotdiv r_1(x_1) )\dotdiv r_2(x_2)$. Here, we think of clients as having a ball of radius zero.

For a point $x\in X_2\cup P$ , let $cl_1(x)$ be the facility $x_1\in X_1$ minimizing $\dists{x_1}{x}$ (we break ties by picking the one minimizing $\dist{x_1}{x}$, and arbitrarily but consistently in case these are still equal).
Similarly $cl_2(x)$ is the $x_2\in X_2$ minimizing $\dists{x}{x_2}$ (breaking ties in the same way).
For a facility $x\in X_1$, let $G_{x}$ denote the set of facilities $x_2$ in $X_2$ such that $cl_1(x_2)=x$.
The sum of radii in $G_{x}$ is $S_{x} = \sum_{x_2 \in G_x} r_2(x_2)$, and the maximum radius in $G_x$ is $M_x = \max_{x_2 \in G_x} r_2(x_2)$.
Finally, let $\Delta(X')$ denote the set of clients $p$ such that $cl_2(p) \in X'$, for any $X'\subseteq X_2$.
Let $r'_1(x_1)= r_1(x_1)+2M_{x_1}$ for all $x_1 \in X_1$. 
Furthermore, let $\distss{x_1}{x}=\left(\dist{x_1}{x}\dotdiv r_1'(x_1)\right)\dotdiv r_2(x)$ for $x_1\in X_1$ and $x\in X_2\cup P$.

The following lemma relates to the above discussion.
It is used to show that for a client $p$, if $cl_2(p)$ is not open in $X$, then there exists a facility $x_1$ from $X_1$ that is open, and the $\distss{\cdot}{\cdot}$ cost is bounded. 
This is a generalized version of triangle inequality.

\begin{lemma} \label{lem:pseudoConnectionCost}
Let $p$ be a client, $x_1=cl_1(p), x_2=cl_2(p),x_1'=cl_1(x_2)$. 
Then
$\distss{x_1'}{p}\le 2\dists{p}{x_2} +\dists{x_1}{p}$.
\end{lemma}
\begin{proof}
Note that $x_2\in G_{x_1'}$ implies $M_{x_1'}\ge r_2(x_2)$ and the definition of $cl_1(\cdot)$ implies $\dists{x_1'}{x_2}\le \dists{x_1}{x_2}$. 

\begin{align*}
&~\distss{x_1'}{p}\\ &= \dist{x_1'}{p} \dotdiv (r_1(x_1')+2M_{x_1'})&\text{(definition of $\delta^{**}$)}\\
&\le \dist{x_1'}{p} \dotdiv (r_1(x_1')+2r_2(x_2)) &\text{($M_{x_1'} \ge r_2(x_2)$)}\\
&\le (\dist{p}{x_2}+ \dist{x_1'}{x_2}) \dotdiv (r_1(x_1')+2r_2(x_2)) &\text{(triangle inequality)}\\
&\le \dist{p}{x_2} \dotdiv r_2(x_2)+ \dist{x_1'}{x_2}\dotdiv (r_2(x_2)+r_1(x_1')) &\text{(property of $\dotdiv$)} \\
&= \dists{p}{x_2} + \dists{x_1'}{x_2} &\text{(definition of $\delta^*$)}\\
&\le \dists{p}{x_2} + \dists{x_1}{x_2} &\text{($\dists{x_1'}{x_2}\le \dists{x_1}{x_2}$)}\\
&= \dists{p}{x_2} + \dist{x_1}{x_2}\dotdiv (r_2(x_2)+r_1(x_1)) &\text{(definition of $\delta^*$)}\\
&\le \dists{p}{x_2} + ( \dist{x_2}{p} + \dist{x_1}{p})\dotdiv (r_2(x_2)+r_1(x_1)) &\text{(triangle inequality)}\\
\end{align*}
\begin{align*}
&\le \dists{p}{x_2} + \dist{x_2}{p} \dotdiv r_2(x_2)+ \dist{x_1}{p}\dotdiv r_1(x_1)&\text{(property of $\dotdiv$)}\\
&= 2\dists{p}{x_2} + \dists{x_1}{p}&\text{(definition of $\delta^*$)}
\end{align*}

\end{proof}

\subsubsection{Constructing the Solution via Linear Programming}\label{subsubsec:constructingvialp}
The idea now is that for $x\in X_1$ we want to either increase the radius by $2M_x$ and open it (this corresponds to opening it with radius $r_1'(x)$) , or open all facilities in $G_{x}$.
Notice that in the second case, if $p\in \Delta(G_{x})$, we pay $\dists{cl_1(p)}{p}+\dists{p}{cl_2(p)}$ less for its connection compared to the upper bound from~\Cref{lem:pseudoConnectionCost};
additionally, as the cost for opening $x$ is $(r_1(x)+2M_x)\rho \le (r_1(x)+2S_x)\rho$, we save $(r_1(x)+S_x)\rho$ (compared to the $(r_1(x)+2S_x)\rho$ upper bound) for the cost of the radii associated with our facilities.
This motivates the LP in \Cref{fig:knapsackLP}, where we (fractionally) open at most $k$ facilities.
When $u_x=0$ we can think of it as opening the facility in $x$, and when $u_x=1$ we open all facilities in $G_x$.

\begin{figure}[!ht]
\caption{\,LP to decide which facilities from $X_1,X_2$ to open.} \label{fig:knapsackLP}
\vskip -1.5ex\rule{\linewidth}{.5pt}
\begin{maxi}|s|<b>{}{\sum_{x\in X_1} u_x \big(r_1(x)\rho + S_x\rho + \sum_{p\in \Delta(G_x)} (\dists{cl_1(p)}{p}+ \dists{p}{cl_2(p)})\big) }{}{}
    \addConstraint{\sum_{x\in X_1} u_x(|G_x|-1)}{\leq k-|X_1|}{}
    \addConstraint{u_x}{\in [0,1]\quad}{\forall x \in X_1}
\end{maxi}
\rule{\linewidth}{.5pt}
\end{figure}

As was noted in \cite{ola}, this is a knapsack LP, meaning that it has an optimal solution where at most one variable is integral.
If all variables are integral, we do not need the following arguments. 
Let $u_{\tildx}$ be the only fractional variable.
We call $\tildx$ from here on the ``special facility''.
We define $G'_{\tildx}= G_{\tildx} \setminus (T\cap \set{\tildx})$ as the group $G_{\tildx}$ where $\tildx$ is removed if it is one of the guessed facilities $T$.
Then we include $\tildx$ in $X$ and set the radius $r(\tildx)= r_1(\tildx)+2\max_{x\in G'_{\tildx}}r_2(x)$.
Note that this radius differs from $r'_1(\tildx)$ in the case that $\tildx\in T$.
This is necessary because $M_{\tildx}\rho$ may not be bounded by $\eps\OPT_\I$, which is necessary to bound the cost of the special facility.  
We include $\ceil{u_{\tildx} |G'_{\tildx}|}-2$ facilities from $G'_{\tildx}$ uniformly at random (which can be derandomized by greedily picking the facilities that maximize the saving).
For all other $x \in X_1 \setminus \set{\tildx}$:
if $u_x=1$ then we include all facilities from $x'\in G_x$ in $X$ and set the radius $r(x')=r_2(x')$.
Else $u_x=0$, in which case we include $x$ in $X$ and set the radius $r(x)=r'_1(x)$.

The solution we create is $\X=(X,r)$.
\subsubsection{Upper Bounding the Number of Balls}
We show that $\X$ is a valid solution by upper bounding the number of opened facilities.
\begin{lemma} \label{lem:pseudoCardinality}
$|X| \le k$.
\end{lemma}
\begin{proof}
When $u_x=1$ we open $|G_x|$ facilities, and when $u_x=0$ we open $1$ facility.
Therefore, when $x\ne \tildx$ we open $(|X_1|-1) + \sum_{x\in X_1\setminus \set{\tildx}} u_x (|G_x|-1)$ facilities in total.
Regarding the group of $\tildx$: we open $\tildx$ and at most $u_{\tildx} |G_{\tildx}|-1$ facilities from $G_{\tildx}$, therefore at most $u_{\tildx} |G_{\tildx}| $ facilities.

By the constraint of the LP we have $u_{\tildx}(|G_{\tildx}|-1) + \sum_{x\in X_1\setminus \set{\tildx}} u_x (|G_x|-1) + |X_1| \le k$.
But the number of facilities we open is at most 
\begin{align*}
&u_{\tildx} |G_{\tildx}| + \sum_{x\in X_1\setminus \set{\tildx}} u_x (|G_x|-1) + (|X_1|-1) \\
=~ &u_{\tildx}(|G_{\tildx}|-1) + \sum_{x\in X_1\setminus \set{\tildx}} u_x (|G_x|-1) + |X_1| + u_{\tildx} - 1\\
\le~ &k + u_{\tildx} - 1 \le k
\end{align*}
\end{proof}

\subsubsection{Analyzing the Cost}
We bound the cost of $\X$ in multiple steps.
At first, we simply bound the number of opened facilities in the special group.
Then, we analyze the linear program and how it connects to the costs of the non-special groups.
Finally, we use these insights to lower bound the cost of $\X$ by a multiple of $\OPT_\I$.

\paragraph{Bounding the new coefficient of the special group}
In the construction of $\X$ we do not open $u_{\tildx}G_{\tildx}-1$ many facilities from the special group because this value may not be an integer.
To not violate the cardinality constraint we open less facilities, namely $\ceil{u_{\tildx} |G'_{\tildx}|}-2$. 
We analyze the fraction of facilities that remain closed in the special group $G'_{\tildx}$.
We upper bound the ratio $(1-p_{\tildx})/(1-u_{\tildx})$ by three, where $p_{\tildx} =( \ceil{u_{\tildx} |G'_{\tildx}|}-2)/|G'_{\tildx}|$ for the facility of the special group $\tildx\in X_1$, that is $u_{\tildx} \notin\set{0,1}$.
\begin{claim}\label{claim:specialfacility}
    \begin{align*}
        \frac{1-p_{\tildx}}{1-u_{\tildx}} \le 3
    \end{align*}
\end{claim}
\begin{proof}
    If $u_{\tildx}\le \nicefrac{2}{3}$, we have $1-u_{\tildx}\ge \nicefrac{1}{3}$. Thus, we conclude $(1-p_{\tildx})/(1-u_{\tildx})\le 3$ because $p_{\tildx}\in [0,1]$.
    Therefore, we assume from here on $u_{\tildx}>\nicefrac{2}{3}$.
    Let $\nu$ be the integer such that ${\nu/(\nu+1)< u_{\tildx}\le (\nu+1)/(\nu+2)}$.
    This implies $1/(\nu+1)> 1-u_{\tildx}\ge 1/(\nu+2)$.
    Additionally, because the LP solution is optimal, the constraint in the LP should be tight. 
    Hence, we have $\xi = u_{\tildx}(|G_{\tildx}|-1)$ is an integer. 
    So, $1-u_{\tildx} = (|G_{\tildx}|-1-\xi)/(|G_{\tildx}|-1)$.
    Because the numerator and denominator are integers, it follows that $|G_{\tildx}|-1>\nu +1$.
    Thus, $|G_{\tildx}'|\ge \nu+2$.
    \begin{align*}
        \frac{1-p_{\tildx}}{1-u_{\tildx}}\le \frac{1-u_{\tildx}+\frac{2}{|G'_{\tildx}|}}{1-u_{\tildx}} = 1+ \frac{2}{|G'_{\tildx}|(1-u_{\tildx})}\le 1+ \frac{2(\nu+2)}{|G'_{\tildx}|}\le 1+\frac{2(\nu+2)}{\nu+2}\le 3
    \end{align*}

\end{proof}

\paragraph{Analyzing the Linear Program}
In this paragraph we lower bound the value of the optimal solution to the linear program.
Additionally, we upper bound the cost of $\X$ in terms of the linear program.
\begin{lemma} \label{lem:pseudoLowerBoundLP}
The solution $u_x=b$ for all $ x\in X_1$, is feasible for the LP in \Cref{fig:knapsackLP} and has value at least $b\sum_{x\in X_1} \big(r_1(x)\rho + S_x\rho + \sum_{p\in \Delta(G_x)} (\dists{cl_1(p)}{p} + \dists{p}{cl_2(p)})\big)$.
\end{lemma}
\begin{proof}
The value of the LP follows directly from the objective function by using $u_x=b$.
It is also a feasible solution because 
\begin{align*}
    \sum_{x\in X_1} u_x(|G_x|-1) &= b \sum_{x\in X_1} (|G_x|-1) = b (|X_2|-|X_1|) = b|X_2| - (1-a) |X_1| \\
    &= a|X_1| + b|X_2| - |X_1| = k-|X_1|.
\end{align*}
\end{proof}

We bound the cost of $\X$ with respect to the value of the optimal solution to the linear program.
\begin{lemma} \label{lem:pseudoSavings}
Let $u_X$ be the optimal solution for the LP in \Cref{fig:knapsackLP}, and
\[U = \sum_{x\in X_1} u_x \big(r_1(x)\rho + S_x\rho + \sum_{p\in \Delta(G_x)} (\dists{cl_1(p)}{p} + \dists{p}{cl_2(p)})\big) \]
be the optimal value of the LP.
Then 
\begin{align*}
\costz{\X} \le 3\sum_{x\in X_1} \big(r_1(x)\rho + 2S_x\rho + \sum_{p\in \Delta(G_x)} (\dists{cl_1(p)}{p} + 2\dists{p}{cl_2(p)})\big) - 3U + 3\varepsilon \OPT_{\I}
\end{align*}
\end{lemma}
\begin{proof}
Let $x\in X_1$.

If $u_x=0$, then by \Cref{lem:pseudoConnectionCost} the part of $\costz{X}$ related to $x$ and $G_x$ is at most

\begin{align*}
&r_1(x)\rho + 2M_x\rho + \sum_{p\in \Delta(G_x)} (\dists{cl_1(p)}{p} + 2\dists{p}{cl_2(p)}) \\
\le &r_1(x)\rho + 2S_x\rho + \sum_{p\in \Delta(G_x)} (\dists{cl_1(p)}{p} + 2\dists{p}{cl_2(p)})
\end{align*}

As $u_x=0$, this is trivially equal to

\begin{align*}
&r_1(x)\rho + 2S_x\rho + \sum_{p\in \Delta(G_x)} (\dists{cl_1(p)}{p} + 2\dists{p}{cl_2(p)}) \\
- u_x \big(&r_1(x)\rho + S_x\rho + \sum_{p\in \Delta(G_x)} (\dists{cl_1(p)}{p} + \dists{p}{cl_2(p)})\big)
\end{align*}

If $u_x=1$ then the part of $\costz{\X}$ related to $x$ and $G_x$ is 
\[S_x\rho + \sum_{p\in \Delta(G_x)} \dists{p}{cl_2(p)}\]

But this is again

\begin{align*}
&r_1(x)\rho + 2S_x\rho + \sum_{p\in \Delta(G_x)} (\dists{cl_1(p)}{p} + 2\dists{p}{cl_2(p)}) \\
- u_x \big(&r_1(x)\rho + S_x\rho + \sum_{p\in \Delta(G_x)} (\dists{cl_1(p)}{p} + \dists{p}{cl_2(p)})\big)
\end{align*}

We  bound the part of $\costz{\X}$ that is related to any non-special facility by summing over these two observations for all non-special facilities.
\begin{align*}
    &\sum_{x\in X_1:u_x=0}\left(r_1(x)\rho + 2M_x\rho + \sum_{p\in \Delta(G_x)} \distss{x}{p}  \right)+ \sum_{x\in X_1:u_x=1}\left(S_x\rho  + \sum_{p\in \Delta(G_x)} \dists{cl_2(p)}{p}  \right)\\
    &\le \sum_{x\in X_1: u_x\in\set{0,1}}\left(r_1(x)\rho + 2S_x\rho +  \sum_{p\in \Delta(G_x)}(\dists{cl_1(p)}{p} + 2\dists{p}{cl_2(p)}) \right)\\
    &-\sum_{x\in X_1:u_x\in\set{0,1}}u_x\left(r_1(x)\rho + S_x\rho +\sum_{p\in \Delta(G_x)}(\dists{cl_1(p)}{p} + \dists{p}{cl_2(p)})\right)
\end{align*}

Finally, it remains to show that the part of $\costz{\X}$ that is related to $\tildx$ is bounded by
\begin{align*}
    & 3\left(r_1(\tildx)\rho + 2S_{\tildx}\rho +  \sum_{p\in \Delta(G_{\tildx})}(\dists{cl_1(p)}{p} + 2\dists{p}{cl_2(p)})\right)\\
    &-3u_{\tildx}\left(r_1({\tildx})\rho + S_{\tildx}\rho +\sum_{p\in \Delta(G_{\tildx})}(\dists{cl_1(p)}{p} + \dists{p}{cl_2(p)})\right)
    +3\eps \OPT_\I
\end{align*}
Recall that the ratio of opened facilities in $G'_{\tildx}$ is $p_{\tildx} =(\ceil{u_{\tildx} |G'_{\tildx}|}-2)/|G'_{x'}|$.

Notice that for all facilities $x' \in G'_{\tildx}$ it holds that $r_2(x')\rho\le \varepsilon \OPT_\I$;
if this was not true, then $x'\in T$ and, thus, it coincides with some other facility in $X_1$, which means it would not be in $G'_{\tildx}$.

We now pay:
\begin{itemize}
    \item $r_1(\tildx)\rho + 2\max_{x\in G'_{\tildx}}r_2(x)\rho \le r_1(\tildx)\rho + 2\eps \OPT_\I$ for opening $\tildx$.
    \item $p_{\tildx} r_2(x')\rho \le u_{\tildx}r_2(x')\rho$ to open facility $x'\in G'_{\tildx}$ (in expectation).
    \item $(1-p_{\tildx})\left(\dists{cl_1(p)}{p} + 2\dists{p}{cl_2(p)}\right) + p_{\tildx} \dists{p}{cl_2(p)}$ for connecting client $p\in \Delta(G_{\tildx})$ with either $\tildx$ or $cl_2(p)$ (in expectation).
\end{itemize}

Using \Cref{claim:specialfacility} we can bound the part of $\costz{\X}$ that is related to $\tildx$ with respect to $u_{\tildx}$ instead of $p_{\tildx}$.

\begin{align*}
    &r_1(\tildx)\rho + 2\eps \OPT_\I + \sum_{x'\in G'_{\tildx}} u_{\tildx}r_2(x')\rho + \sum_{p\in \Delta(G_{\tildx})} ((1-p_{\tildx})\dists{cl_1(p)}{p} + (2-p_{\tildx})\dists{p}{cl_2(p)})  \\
    \le~&r_1(\tildx)\rho + 2\eps \OPT_\I + \sum_{x'\in G'_{\tildx}} u_{\tildx}r_2(x')\rho + 3\sum_{p\in \Delta(G_{\tildx})} ((1-u_{\tildx})\dists{cl_1(p)}{p} + (2-u_{\tildx})\dists{p}{cl_2(p)})
\end{align*}

We distinguish two cases.
In the first case we assume $\tildx$ is one of the guessed facilities in $T$; then it is also in $X_2$.
By the definition of the groups, it is also in its group $G_{\tildx}$, but not in $G'_{\tildx}$.
In the second case we assume $\tildx$ is none of the guessed facilities; then it has small radius.

We start with the first case.
If $\tildx \in T$, we have:
\begin{align*}
&r_1(\tildx)\rho + 2\eps \OPT_\I + \sum_{x'\in G'_{\tildx}} u_{\tildx}r_2(x')\rho \\
&\qquad +3\sum_{p\in \Delta(G_{\tildx})} ((1-u_{\tildx})\dists{cl_1(p)}{p} + (2-u_{\tildx})\dists{p}{cl_2(p)})  \\
=~&(1-u_{\tildx})r_1(\tildx)\rho + 2\eps \OPT_\I + u_{\tildx}S_{\tildx}\rho \\
& \qquad+3\sum_{p\in \Delta(G_{\tildx})} ((1-u_{\tildx})\dists{cl_1(p)}{p} + (2-u_{\tildx})\dists{p}{cl_2(p)}  \\
\le~&(1-u_{\tildx})r_1(\tildx)\rho + 2\eps \OPT_\I + (2-u_{\tildx})S_{\tildx}\rho \\
&\qquad + 3\sum_{p\in \Delta(G_{\tildx})} ((1-u_{\tildx})\dists{cl_1(p)}{p} + (2-u_{\tildx})\dists{p}{cl_2(p)}  \\
\le~&3\left(r_1(\tildx)\rho + 2S_{\tildx}\rho + \sum_{p\in \Delta(G_{\tildx})} (\dists{cl_1(p)}{p} + 2\dists{p}{cl_2(p)})\right) \\
&\qquad- 3u_{\tildx} \left(r_1(\tildx)\rho + S_{\tildx}\rho + \sum_{p\in \Delta(G_{\tildx})} (\dists{cl_1(p)}{p} + \dists{p}{cl_2(p)})\right) + 2\eps \OPT_\I
\end{align*}

On the other hand, if $\tildx\not \in T$, then $r_1(\tildx) \rho \le \eps \OPT_\I$, therefore our cost is upper bounded by

\begin{align*}
&3\eps \OPT_\I + u_{\tildx}S_{\tildx}\rho + 3\sum_{j\in \Delta(G_{\tildx})} ((1-u_{\tildx})\dists{cl_1(p)}{p} + (2-u_{\tildx})\dists{p}{cl_2(p)}) \\
\le~&3\eps \OPT_\I + (2-u_{\tildx})S_{\tildx}\rho + 3\sum_{p\in \Delta(G_{\tildx})} ((1-u_{\tildx})\dists{cl_1(p)}{p} + (2-u_{\tildx})\dists{p}{cl_2(p)})  \\
\le~&3\left(r_1(\tildx)\rho + 2S_{\tildx}\rho + \sum_{p\in \Delta(G_{\tildx})} (\dists{cl_1(p)}{p} + 2\dists{p}{cl_2(p)})\right) \\
&\quad- 3u_{\tildx} \left(r_1(\tildx)\rho + S_{\tildx}\rho + \sum_{p\in \Delta(G_{\tildx})} (\dists{cl_1(p)}{p} + \dists{p}{cl_2(p)})\right) + 3\eps \OPT_\I
\end{align*}
\end{proof}

\paragraph{Bounding the cost of $\X$}
We are now ready to bound the cost of $\X$ in terms of the cost $\OPT_\I$ of the optimal solution.

\begin{lemma} \label{lem:pseudoApx}
$\costz{\X} \le (13.5+7.5\eps)\OPT_\I$
\end{lemma}
\begin{proof}Let $U$ be the optimal value for the LP in \Cref{fig:knapsackLP}.
By \Cref{lem:pseudoSavings} we have 

\begin{align*}
\costz{\X} \le 3\sum_{x\in X_1} \big(r_1(x)\rho + 2S_x\rho + \sum_{p\in \Delta(G_x)} (\dists{cl_1(p)}{p} + 2\dists{p}{cl_2(p)})\big) - 3U +3\eps\OPT_\I
\end{align*}

Then by \Cref{lem:pseudoLowerBoundLP} we get that

\begin{align*}
\costz{\X} \le~&3\sum_{x\in X_1} \big(r_1(x)\rho + 2S_x\rho + \sum_{p\in \Delta(G_x)} (\dists{cl_1(p)}{p} + 2\dists{p}{cl_2(p)})\big)\\
&-3b\sum_{x\in X_1} \big(r_1(x)\rho + S_x\rho + \sum_{p\in \Delta(G_x)} (\dists{cl_1(p)}{p} + \dists{p}{cl_2(p)})\big) + 3\eps \OPT_\I
\end{align*}

But as $a+b=1$, we get

\begin{align*}
\costz{\X} &\le 3\sum_{x\in X_1} \big(a\cdot r_1(x)\rho + (1+a)S_x\rho + \sum_{p\in \Delta(G_x)} (a\cdot \dists{cl_1(p)}{p} + (1+a)\dists{p}{cl_2(p)})\big) + 3\eps \OPT_\I\\
&=3a \cdot \costz{\X_1} + 3(1+a)\costz{\X_2} + 3\eps \OPT_\I
\end{align*}

Now since $a\le b$ and $a \cdot \costz{\X_1}+b \cdot \costz{\X_2} \le (3+\eps)\OPT_\I$:

\[\costz{\X} \le (9+6\eps)\OPT_\I + 6a\costz{\X_2} \]

Recall that $a\le \nicefrac{1}{4}$. Finally, we have $\costz{\X_2} \le (3+\eps)\OPT_\I \le \costz{\X_1}$, meaning

\[\costz{\X} \le (13.5+7.5\eps)\OPT_\I \]
\end{proof}
Now our main result follows.

\begin{proof}[Proof of~\Cref{thm:ballkmed}]
Follows directly by \Cref{lem:pseudoCardinality} and \Cref{lem:pseudoApx}.
\end{proof}

\apxtoplone*
\begin{proof}
    We have an approximation for \Ballk{} due to \Cref{thm:ballkmed}.

    Due to the approximation preserving reduction in \Cref{lem:redball} there is also a factor-$(13.5+\eps)$ approximation for \NCC{\textsf{Top}}{\LP{1}}.
\end{proof}

\section{\texorpdfstring{\NCC{\LP{\infty}}{\textsf{Ord}}}{\NCCH{\textnormal{l inf}}{\textsf{Ord}}}}

In this section we show an $O(1)$-approximation for \NCC{\LP{\infty}}{\textsf{Ord}}. 
Generally, a solution to \NNCC{} contains an assignment $\sigma:P \rightarrow X$. 
But similarly to the \msr{} problem, the cost of a cluster around $x$ is just the largest distance $\dist{x}{p}$ to an assigned client $p\in P$.
Therefore, the solution to \NNCC{} is a set of balls around facilities, that cover all the clients $P$.
The goal is to minimize the ordered norm of the radii. 
We formalize this equivalent formulation of \NCCS{\LP{\infty}}{\textsf{Ord}}.

For a set of facilities $X$ and a radius function $r:X\rightarrow \nnr$ let $\bm{r}(X)$ be the vector of the radii of the facilities in $X$ in arbitrary order.
Furthermore let $B^r(x)=\{p\in P\mid \dist{x}{p}\le r(x)\}$ be the clients covered by a facility $x$ and $B^r(X)=\bigcup_{x\in X}B^r(x)$ for a set of facilities $X$.
\begin{Definition}[\textnormal{\NCCS{\LP{\infty}}{\textsf{Ord}}}]
     An instance $\I = (P,F,\delta,k,\bm{w})$ consists of a point set $P$, a set $F$ of facilities, a metric $\delta: (P\cup F) \times (P\cup F)\rightarrow \nnr$, a $k\in \mathbb{N}$ and a non-increasing weight vector $\bm{w}\in \mathbb{R}^k_{\ge 0}$. 
    The goal is to find a solution $\X=(X,r)$ containing a set $X\subseteq F$ of facilities of size $|X|=k$ with radii $r:X\rightarrow \mathbb{R}^{\ge0}$ such that $B^r(X)=P$ which minimizes 
    $\ord{\bm{w}}{\bm{r}(X)}$.
\end{Definition}

\subsection{Reduction to \MSRDC}
In this section we reduce \NCCS{\LP{\infty}}{\textsf{Ord}} to \MSRDC{} (\MSRDCS).
Intuitively, this is a problem similar to \msr{}, but instead of paying for the sum of radii, we pay for the sum of an increasing function of the radii.

Let us begin with the formal definition of the problem.
\begin{Definition}[\MSRDC{} (\MSRDCS)]
    An instance $\I = (P,F,\delta,k,h)$ consists of a point set $P$, a set $F$ of facilities, a metric $\delta: (P\cup F) \times (P\cup F)\rightarrow \nnr$, a $k\in \mathbb{N}$ and a non-decreasing cost function $h:\nnr\rightarrow \nnr$. 
    The goal is to find a solution $\X=(X,r)$ containing a set $X\subseteq F$ of facilities of size $|X|=k$ with radii $r:X\rightarrow \nnr$ such that $B^r(X)=P$ which minimizes 
    \begin{align*}
        h_r(X) = \sum_{x\in X}h(r(x)).
    \end{align*}
\end{Definition}

At first, we want to reduce $\dif{\bm{w}}=|\{i\in [k-1]\mid \bm{w}_i\ne \bm{w}_{i+1}\}|$, which measures the complexity of the ordered norm with weight vector $\bm{w}$. 

\begin{lemma}[Lemma 4.2 in \cite{chakrabarty-swamy19:norm-k-clustering}]\label{lem:reducedcomplex}
    For all $\bm{w}\in \mathbb{R}^k_{\ge0}$ and $\eps >0$ there is a $\bm{\Tilde{w}}\in \mathbb{R}^k_{\ge0}$ such that $\dif{\tildw}\le O(\log k)$ and for all $\bm{x}\in \mathbb{R}^k_{\ge0}$ it holds that $\ord{\tildw}{\bm{x}}\le \ord{\bm{w}}{\bm{x}}\le (1+\eps) \ord{\tildw}{\bm{x}}$. 
    Additionally, we can compute $\tildw$ in polynomial time.
\end{lemma}

Furthermore, we can compute a threshold vector, that is close enough to the cluster cost vector of the optimal solution, in polynomial time.

\begin{lemma}[Lemma 6.8 and 6.9 in \cite{chakrabarty-swamy19:norm-k-clustering}]\label{lem:nearoptimal}
    
    For all $\eps>0$ and instances $\I=(P,F,\delta,k,\bm{w})$ with $\dif{\bm{w}}\le O(\log k)$, in polynomial time we can compute a set of threshold vectors $A$ containing a threshold vector $\bm{t}$ with $\proxyz{\bm{r}^*_\I(X^*_\I)}{\bm{w}}{\bm{t}}\le (1+\eps)\OPT_\I$.
\end{lemma}

We are now ready to present our reduction:

\begin{lemma}\label{lem:reducetogeneralmsr}
    Assume there is an algorithm that computes for an instance $\I=(P,F,\delta,k,h)$ of \MSRDCS{} a solution $\X= (X,r)$ such that 
    \begin{align*}
        \sum_{x\in X}h\left(\frac{r(x)}{9}\right) \le (2+2\eps')\OPT_\I
    \end{align*}
    for all $\eps'>0$, then there is a factor-$(18+\eps)$ approximation for \NCCS{\LP{\infty}}{\textnormal{\textsf{Ord}}} for all $\eps>0$.
\end{lemma}

\begin{proof}
    Let $\I'=(P,F,\delta,k,\bm{w})$ and $\eps'>0$ such that $(1+\eps')^3\le 1+\eps$.
    We use~\Cref{lem:reducedcomplex} and~\Cref{lem:nearoptimal} to find a weight vector $\bm{\Tilde{w}}\in \mathbb{R}^k_{\ge 0}$ and a threshold vector $\bm{t}\in \mathbb{R}^k_{\ge0}$, such that 
    \begin{enumerate}
        \item $\proxy{\bm{t}}{\bm{r}^*_\I(X^*_\I)}{\bm{\tildw}}\le (1+\eps')\OPT_{\I'}$ and
        \item $(1+\eps')\ord{\tildw}{\bm{x}} \ge \ord{\bm{w}}{\bm{x}}$ for all $\bm{x}\in \mathbb{R}^k_{\ge 0}$.
    \end{enumerate}

Assume we have an approximation algorithm as in the statement of the lemma.
Let $h(a) = \sum_{i=1}^{k} (\bm{\tildw}_i-\bm{\tildw}_{i+1})(a\dotdiv \bm{t}_i)$.
    We show how to obtain a $(18+\eps)$-approximation for the objective function $\proxy{\bm{\tildw}}{\bm{r}(X)}{\bm{t}}$. 
    Let $\X=(X,r)$ be the solution of the approximation algorithm.
    \begin{align*}
        18 (1+\eps')^3\OPT_{\I'} &\ge 18(1+\eps')^2 \proxyz{\bm{r}_{\I'}^*(X^*_{\I'})}{\bm{\tildw}}{\bm{t}}\\
        &\ge 18(1+\eps')^2\sum_{x\in X^*_{\I'}}h(r_{\I'}^*(x))+18(1+\eps')\sum_{i=1}^k i(\bm{\tildw}_i-\bm{\tildw}_{i+1})\bm{t}_i\\
        &\ge 18(1+\eps')^2\sum_{x\in X^*_{\I}}h(r_{\I'}^*(x))+18(1+\eps')\sum_{i=1}^k i(\bm{\tildw}_i-\bm{\tildw}_{i+1})\bm{t}_i\\
        &\ge  9(1+\eps')\sum_{x\in X}h\left(\frac{r(x)}{9}\right)+9(1+\eps')\sum_{i=1}^k i(\bm{\tildw}_i-\bm{\tildw}_{i+1})\bm{t}_i\\
        &=  9(1+\eps')\sum_{x\in X}\sum_{i=1}^{k} (\bm{\tildw}_i-\bm{\tildw}_{i+1})\left(\frac{r(x)}{9}\dotdiv \bm{t}_i\right)+9(1+\eps')\sum_{i=1}^k i(\bm{\tildw}_i-\bm{\tildw}_{i+1})\bm{t}_i\\ 
        &\ge (1+\eps') \sum_{x\in X}\sum_{i=1}^{k} (\bm{\tildw}_i-\bm{\tildw}_{i+1})\left(r(x)\dotdiv 9\bm{t}_i\right)+(1+\eps')\sum_{i=1}^k i(\bm{\tildw}_i-\bm{\tildw}_{i+1})9\bm{t}_i\\ 
        &=(1+\eps')\proxyz{\bm{r}(X)}{\tildw}{\bm{9t}}\ge (1+\eps')\ord{\tildw}{\bm{r}(X)} \ge \ord{\bm{w}}{\bm{r}(X)}
    \end{align*}       
\end{proof}

\subsection{Approximating \MSRDC}
By \Cref{lem:reducetogeneralmsr}, it suffices to approximate \MSRDC{}.
We start with a straightforward observation, namely that it only makes sense to consider radii equal to distances between facilities and clients.

\begin{observation}\label{obs:realdist}
    For any instance $\I = (P,F,\delta,k,h)$ of \MSRDCS{} and for all solutions $(X,r)$ to $\I$ there is a solution $(X,r')$ such that $r'(x)\in R_x$ for all $x\in X$ and \begin{align*}
        h_{r'}(X)\le h_r(X),
    \end{align*}
    where $R_x = \{\dist{x}{p}\mid p\in P\}$.
\end{observation}

\subsubsection{Guessing the Largest Balls}\label{subsubsec:guessz}
As in \Cref{sec:topl}, a technical nuance of our analysis is that we are not able to bound the cost regarding one ``special'' ball.
For fix that, we first guess the largest balls from the optimal solution, ensuring that in the reduced instance it suffices to focus on low-cost balls.

Algorithm~\ref{alg:guessoptz} guesses the $O(\nicefrac{1}{\eps})$ facilities (and their respective radii) with the largest radii, from the optimal solution $\X^*_\I=(X^*_\I,r^*_\I)$ of instance $\I$.
We assume without loss of generality that for all $x\in X^*_\I$, we have $r(x)=\dist{p}{x}$ for some $p\in P$ by~\Cref{obs:realdist}.

\begin{algorithm2e}[!ht]
  \SetKwFunction{Guess}{GUESS}

  \setcounter{AlgoLine}{0}
  \SetKwProg{procedure}{Procedure}{}{}
  
   \procedure{\Guess{$\I = (P,F,\delta,k,\ell), \eps$}}{
  $t\gets \ceil{\nicefrac{1}{\eps}}$\;
  \ForEach{$T\in \binom{F}{t}$ and $r\colon T\rightarrow \{\dist{p}{x}\mid p\in P, x\in T\} $}{
    \textbf{output} $(T,r)$\;
    }
  }

\caption{Guessing the largest balls.}
\label{alg:guessoptz}

\end{algorithm2e}

We claim that at least one of the pairs that Algorithm~\ref{alg:guessoptz} outputs is a correct guess.
Therefore, all radii of facilities outside of the guessed facilities have small radii. 
We omit the proof as it is very similar to the proof of~\Cref{lem:guessopt}.

\begin{lemma}\label{lem:guessoptz}
    Algorithm~\ref{alg:guessoptz} outputs in time $n^{O(\nicefrac{1}{\eps})}$ a list of pairs so that at least one pair, say $(T,r)$, satisfies
    \begin{itemize}
        \item $T\subseteq X^*_{\I}$ and for all $x\in T$ it holds that $r(x)=r^{*}_\I(x)$,
        \item For all $x\in X^*_\I\setminus T$ it holds that $r^*_\I(x) \le  \min_{x'\in T}  r(x')$
        \item $\min_{x'\in T} h(r(x')) \le \eps\cdot\OPT_\I$. 
    \end{itemize}
\end{lemma}

\subsubsection{LMP Approximation}
In this section we define the facility location version of \MSRDCS{} and show a special type of approximation algorithm for it.

\begin{Definition}[\FLMSRDCS{}]
    An instance $\I = (P,F,\delta,h,\lambda)$ consists of a point set $P$, a set $F$ of facilities, a metric $\delta: (P\cup F) \times (P\cup F) \rightarrow \nnr$, a non-decreasing cost function $h:\nnr\rightarrow \nnr$ and an opening cost $\lambda \ge 0$. 
    The goal is to find a solution $\X=(X,r)$ containing a set $X\subseteq F$ of facilities with radii $r:X\rightarrow \nnr$ such that $B^r(X)=P$ which minimizes 
    \begin{align*}
       h_r(X)+|X|\lambda = \sum_{x\in X}h(r(x)) + |X|\lambda.
    \end{align*}
\end{Definition}

\Cref{fig:MSRLP} shows an LP-relaxation of \MSRDCS{} where the distances are bounded by an additional parameter $\mu \ge 0$, that is $R_x^\mu =\{\dist{x}{p}\mid p\in P \wedge h(\dist{x}{p})\le \mu\}$.
Note that for $\mu \ge \max_{x\in X^*_\I}h(r^*_\I(x))$ the optimal solution $\X^*_\I$ is a valid solution to the LP.
\Cref{fig:dualmsr} show the corresponding dual-LP. 
They are the basis for Algorithm~\ref{alg:lmpmsr}, which works as follows:
    \begin{enumerate}
        \item Increase all $\alpha_p$ simultaneously at the same rate. When a constraint gets tight, that is $\sum_{\dist{x}{p}\le r}\alpha_p = h(r)+\lambda$ for some $x\in F,r\in R_x^\mu$, stop increasing the corresponding $\alpha_p$s, open facility $x$ temporarily and set $r'(x)\gets \max\{r'(x),r\}$.
        Stop when no $\alpha_p$s are increasing.
        \item Let $Y$ be the set of temporarily opened facilities. Consider the graph $G$ on $Y$ where $x_1,x_2\in Y$ are connected iff $B^{r'}(x_1)\cap B^{r'}(x_2) \ne \emptyset$.
        \item Construct a maximum independent set $X$ by greedily selecting the temporarily opened facility with the largest radius $r'(x)$.
        \item Set $r(x)\gets 3r'(x)$ for all $x\in X$ and return $(X,r)$.
    \end{enumerate}

\begin{figure}[!ht]
\caption{\,LP for \FLMSRDCS{}.} \label{fig:MSRLP}
\vskip -1.5ex\rule{\linewidth}{.5pt}
\begin{mini}|s|<b>{}{\sum_{x\in F,r\in R^\mu_x}u_x^r(\lambda+h(r))}{}{}
    \addConstraint{\sum_{x\in F,r\in R^\mu_x:\dist{p}{x}\le r}u_x^r }{\geq 1\qquad}{\forall p \in P}
\end{mini}
\rule{\linewidth}{.5pt}
\end{figure}

\begin{figure}[!ht]
\caption{\,Dual-LP for Figure~\ref{fig:MSRLP}.} \label{fig:dualmsr}
\vskip -1.5ex\rule{\linewidth}{.5pt}
\begin{maxi}|s|<b>{}{\sum_{p\in P} \alpha_p}{}{}
    \addConstraint{\sum_{p\in P:\dist{p}{x}\le r}\alpha_{p}}{\leq \lambda+h(r)\qquad}{\forall x \in F, r\in R^\mu_x}
\end{maxi}
\rule{\linewidth}{.5pt}
\end{figure}

\begin{algorithm2e}[!ht]
  \SetKwFunction{Msrdc}{MSRDCF}

  \setcounter{AlgoLine}{0}
  \SetKwProg{procedure}{Procedure}{}{}
  
   \procedure{\Msrdc{$\I = (P,F,\delta,h,\lambda), \mu$}}{
  $Y\gets \emptyset$\;
  $r'(x)\gets 0$ for all $x\in F$\;
  $\alpha_p \gets 0$ for all $p\in P$\;
  Start increasing all $\alpha_p$ simultaneously at the same rate\;
  \While{$\alpha_p$ is increasing for some $p\in P$}{
    \If{$\sum_{p\in P:\dist{p}{x}\le r}\alpha_{p}= \lambda + h(r)$ for some $x\in F,r\in R_x^\mu$}{
        $Y\gets Y\cup \{x\}$\;
        $r'(x)\gets \max\{r'(x), r\}$\;
        Stop increasing $\alpha_p$ for $p$ such that $\dist{p}{x}\le r$\;
    }
    }
  }
  $X\gets \emptyset$\;
  \While{$Y\ne \emptyset$}{
  Pick $x = \arg\max_{x\in Y}r'(x)$\;
  $X\gets X\cup \{x\}$\;
  $r(x)\gets 3r'(x)$\;
  $Y\gets Y \setminus \{x'\in Y\mid B^{r'}(x)\cap B^{r'}(x')\ne\emptyset\}$\;
  }
  \Return $(X,r)$\;

\caption{Approximate \FLMSRDCS{}.}
\label{alg:lmpmsr}

\end{algorithm2e}

\begin{lemma}\label{lem:LMP}
    Given an instance $\I=(P,F,\delta,k,h)$ of \MSRDCS, $\mu \ge \max_{x\in X^*_\I}h(r^{*}_\I(x))$ and $\lambda \ge 0$, Algorithm~\ref{alg:lmpmsr} returns a  pair $\X=(X,r)$ in polynomial time such that
    \begin{align*}
        \sum_{x\in X}h\left(\frac{r(x)}{3}\right)+|X|\lambda \le h(\bm{r}^*_\I(X^*_\I)) + k\lambda
    \end{align*}
    and for all facilities $x\in X$ we have $h(\frac{r(x)}{3})\le \mu$.
\end{lemma}
\begin{proof}
    We start by showing that all radii in $X$ are small, that is $h(r(x)/3)\le \mu$ for all $x\in X$. 
    As for all $x\in X$ we have $r'(x) \in R_x^\mu$, then $h(r'(x))\le \mu$. 
    Since the radii are tripled in the construction of $X$, we conclude $h(r(x)/3)\le \mu$ for all $x\in X$.

    We now bound the cost of $X$.
    As we choose the facilities in $X$ to be independent, we know 
    \begin{align*}
        |X|\lambda +\sum_{x\in X}h\left(\frac{r(x)}{3}\right)= |X|\lambda +\sum_{x\in X}h(r'(x))\le \sum_{x\in X}\sum_{p\in B^{r'}(x)}\alpha_p\le \sum_{p\in P}\alpha_p\le h(\bm{r}^*_\I(X^*_\I))+k\lambda.
    \end{align*}
    
    In the remaining part of the proof we show that $B^r(X) = P$. Consider an arbitrary client $p\in P$.
    By the definition of $Y$ there is an $x\in Y$ such that $p\in B^{r'}(x)$. 
    If $x\notin X$, then there is a $x' \in X$ and a $p'\in P$ such that $p'\in B^{r'}(x)\cap B^{r'}(x')$ and $r'(x')\ge r'(x)$. 
    We use triangle inequality to bound $\dist{p}{x'}$.
    \begin{align*}
        \dist{p}{x'}\le \dist{p}{x}+\dist{x}{p'}+\dist{p'}{x'}\le 2r'(x)+r'(x')\le 3r'(x') = r(x')
    \end{align*}
    Thus, $p\in B^r(x)$.
\end{proof}

As is standard, we use the LMP approximation to obtain a bi-point solution, which we then use to obtain our approximation.

\begin{lemma}\label{lem:binarysear}
    Given an instance $\I=(P,F,\delta,k,h)$ of \MSRDCS{}, $\eps >0$, and $\mu \ge \max_{x\in X^*_\I}h(r^{*}_\I(x))$, we can compute a bi-point solution $(\X_1=(X_1,r_1),\X_2=(X_2,r_2))$ of cost
    \begin{align*}
        a\cdot \sum_{x_1\in X_1}h\left(\frac{r_1(x_1)}{3}\right)+b\cdot \sum_{x_2\in X_2}h\left(\frac{r_2(x_2)}{3}\right)\le (1+\eps) \OPT_\I
    \end{align*}
    such that $a,b\ge 0, a+b=1, a|X_1|+b|X_2|=k$, for all $x\in X_1$ it holds that $h(\frac{r_1(x)}{3})\le \mu$ and for all $x\in X_2$ it holds that $h(\frac{r_2(x)}{3})\le \mu$.
\end{lemma}
This is proven similar to~\Cref{lem:bipoint}, using \Cref{lem:LMP}.

\begin{lemma}\label{lem:bipointround}
    Given an instance $\I=(P,F,\delta,k,h)$ of \MSRDCS{} and $\eps >0$, we can compute a solution $\X=(X,r)$ such that 
    \begin{align*}
        \sum_{x\in X}h\left(\frac{r(x)}{9}\right)\le (2+3\eps)\OPT_\I.
    \end{align*}
\end{lemma}
\begin{proof}
We first use \Cref{lem:guessoptz} to guess the largest $\lceil \nicefrac{1}{\eps} \rceil$ balls in the optimal solution; let $r_m$ be the minimum out of these radii, and $\mu = h(r_m)$.
Let $P'\subseteq P$ be the clients not covered by these balls and $k'=k-\lceil \nicefrac{1}{\eps} \rceil$.
    We are thus left with solving the reduced instance 
    $\I'=(P',F,\delta,k',h)$.
    Furthermore, by \Cref{lem:guessoptz} it follows that $\mu \in [\max_{x\in X^*_{\I'}}h(r^{*}_{\I'}(x)), \eps OPT_{\I}]$.

    We can now use \Cref{lem:binarysear} to obtain a bi-point solution $(\X_1=(X_1,r_1),\X_2=(X_2,r_2))$ of cost

    \begin{align*}
        a\cdot \sum_{x_1\in X_1}h\left(\frac{r_1(x_1)}{3}\right)+b\cdot \sum_{x_2\in X_2}h\left(\frac{r_2(x_2)}{3}\right)\le (1+\eps) \OPT_{\I'}
    \end{align*}
    such that $a,b\ge 0, a+b=1, a|X_1|+b|X_2|=k'$, for all $x\in X_1$ it holds that $h(\frac{r_1(x)}{3})\le \eps\OPT_\I$ and for all $x\in X_2$ it holds that $h(\frac{r_2(x)}{3})\le \eps\OPT_\I$.

    Assume, w.l.o.g. $\sum_{x_1\in X_1}h\left(\frac{r_1(x_1)}{3}\right) \ge \sum_{x_2\in X_2}h\left(\frac{r_2(x_2)}{3}\right)$.
    For a facility $x\in X_2$ , let $cl_1(x)$ be the facility $x_1\in X_1$ such that $B^{r_1}(x_1)\cap B^{r_2}(x)\ne \emptyset$ (we break ties by picking the one minimizing $\dist{x}{x_1}$, and arbitrarily but consistently in case these are still equal).
For a facility $x\in X_1$, let $G_{x}$ denote the set of facilities $x_2$ in $X_2$ such that $cl_1(x_2)=x$.
Let $\Tilde{X_1}$ be the set of facilities $x_1$ such that $G_{x_1}\ne \emptyset$.
The maximum radius in $G_x$ is $M_x = \max_{x_2 \in G_x} r_2(x_2)$.

We consider two options of covering the clients in $B^{r_2}(G_{x_1})$ for $x_1\in\Tilde{X_1}$. Either we open $x_1$ with radius $r(x_1)=r_1(x_1)+2M_{x_1}$ or we open all facilities in $G_{x_1}$ with radius $r(x_2) = r_2(x_2)$ for all $x_2\in G_{x_1}$.

This motivates the LP in~\Cref{fig:knapsackLPz}, where we (fractionally) open at most $k$ facilities. When $u_{x_1}=0$ we can think of it as opening the facility $x_1$ and when $u_{x_1}=1$ we open all facilities in $G_{x_1}$.
The intuition behind the objective function is that it represents the amount saved, compared to the solution that opens all facilities from $\tilde{X_1}$.
Since this is a knapsack LP, it has an optimal solution where all but one variables are integral. 
Let $u_{\tilde{x}}$ be the only fractional variable.
Then, we include $\tilde{x}$ in $X$.

For all other $x_1\in \Tilde{X_1} \setminus \{\tilde{x}\}$: if $u_{x_1}=1$ then we include all facilities from $G_{x_1}$ in $X$. 
Else $u_{x_1}=0$ and we include $x_1$ in $X$.

    \begin{figure}[!ht]
\caption{\,LP to decide which facilities from $\tilde{X_1},X_2$ to open.} \label{fig:knapsackLPz}
\vskip -1.5ex\rule{\linewidth}{.5pt}
\begin{maxi}|s|<b>{}{\sum_{x_1\in \Tilde{X_1}} u_{x_1}\left(h\left(\frac{r_1(x_1)+2M_{x_1}}{9}\right)-\sum_{x_2\in G_{x_1}}h\left(\frac{r_2(x_2)}{9} \right)\right) }{}{}
    \addConstraint{\sum_{x_1\in \Tilde{X_1}} u_{x_1}(|G_{x_1}|-1)}{\leq k'-|\Tilde{X_1}|}{}
    \addConstraint{u_{x_1}}{\in [0,1]\quad}{\forall x_1 \in \Tilde{X_1}}
\end{maxi}
\rule{\linewidth}{.5pt}
\end{figure}
    \begin{claim}
        $|X|\le k'$
    \end{claim}
    \begin{proof}
        When $u_{x_1}=1$ we open $|G_{x_1}|$ facilities and when $u_{x_1}=0$ we open one facility.
        Therefore, when $x_1\ne \tilde{x}$ we open $\sum_{x_1\in \Tilde{X_1}\setminus \{\tilde{x}\}}u_{x_1}(|G_{x_1}|-1)+(|\Tilde{X_1}|-1)$ facilities in total.
        Regarding the group of $\tilde{x}$: we open at most one facility.

        We conclude $\sum_{x_1\in \Tilde{X_1}\setminus \{\tilde{x}\}}u_{x_1}(|G_{x_1}|-1) + (|\Tilde{X_1}|-1) + 1 \le (k'-|\Tilde{X_1}|) + |\Tilde{X_1}| = k'$
    \end{proof}
    \begin{claim}
        The optimal solution to the LP has value
        \begin{align*}
            U \ge b\sum_{x_1\in \Tilde{X_1}}\left(h\left(\frac{r_1(x_1)+2M_{x_1}}{9}\right)-\sum_{x_2\in G_{x_1}}h\left(\frac{r_2(x_2)}{9} \right)\right) .
        \end{align*}
    \end{claim}
    \begin{proof}
        It suffices to show that setting $u_{x_1}=b$ for all $x_1$ is a valid solution to the LP.
        
        We have $\sum_{x\in \tilde{X_1}} b(|G_x|-1) = b (|X_2|-|\tilde{X_1}|) = b|X_2| - (1-a) |\tilde{X_1}| \le a|X_1| + b|X_2| - |\tilde{X_1}| = k'-|\tilde{X_1}|$.
        Also $b\in [0,1]$, which concludes the proof.
    \end{proof}

    We now analyze the cost of $(X,r)$.
    \begin{align*}
        \sum_{x\in X}h\left(\frac{r(x)}{9}\right) \le& \sum_{x_1\in \Tilde{X_1}}h\left(\frac{r_1(x_1)+2M_{x_1}}{9}\right) - U + h\left(\frac{r(\tilde{x})}{9}\right)\\
        \le& \sum_{x_1\in \Tilde{X_1}}h\left(\frac{r_1(x_1)+2M_{x_1}}{9}\right) - U + \eps OPT_\I\\
        \le& \sum_{x_1\in \Tilde{X_1}}h\left(\frac{r_1(x_1)+2M_{x_1}}{9}\right) - b\sum_{x_1\in \Tilde{X_1}} \left(h\left(\frac{r_1(x_1)+2M_{x_1}}{9}\right)-\sum_{x_2\in G_{x_1}}h\left(\frac{r_2(x_2)}{9} \right)\right) + \eps \OPT_\I\\
        =& \sum_{x_1\in \Tilde{X_1}}\left(a\cdot h\left(\frac{r_1(x_1)+2M_{x_1}}{9}\right) + b \cdot\sum_{x_2\in G_{x_1}} h\left(\frac{r_2(x_2)}{9} \right)\right) + \eps \OPT_\I\\
        \le& \sum_{x_1\in \Tilde{X_1}}\left(a\cdot \left(h\left(\frac{3r_1(x_1)}{9}\right)+ h\left(\frac{3M_{x_1}}{9}\right)\right)+ b \cdot\sum_{x_2\in G_{x_1}} h\left(\frac{r_2(x_2)}{9} \right)\right)  +\eps \OPT_\I\\
        \le& \sum_{x_1\in \Tilde{X_1}}\left(a\cdot h\left(\frac{r_1(x_1)}{3}\right)+  \sum_{x_2\in G_{x_1}} h\left(\frac{r_2(x_2)}{3} \right)\right) + \eps \OPT_\I\\
        \le~& a\cdot\sum_{x_1\in \Tilde{X_1}} h\left(\frac{r_1(x_1)}{3}\right)+  \sum_{x_2\in X_2} h\left(\frac{r_2(x_2)}{3} \right) +\eps \OPT_\I\\
        \le~& (2+2\eps)\OPT_{\I'}+\eps\OPT_\I \\
    \end{align*}
As the cost for the guessed balls is $OPT_\I - OPT_{\I'}$, the result follows.
\end{proof}

\thmlinford*
\begin{proof}
    This is a direct consequence of Lemmas~\ref{lem:reducetogeneralmsr}~and~\ref{lem:bipointround}.
\end{proof}

\section{\texorpdfstring{\NCC{\LP{\infty}}{\textsf{Sym}}}{\NCCH{\textnormal{l inf}}{\textsf{Sym}}}}

In this section we design an $(O(1),O(1))$-bicriteria approximation algorithm for \NCCS{\LP{\infty}}{\textsf{Sym}}, that is an $O(1)$ approximation opening $O(k)$ many facilities.

\begin{theorem}\label{thm:bicriteria}
    There is a universal constant $\alpha >0$ such that given an instance $\I=(P,F,\delta,k,\LP{\infty},g)$ of \NCCS{\LP{\infty}}{\textsf{Sym}} and a $\kappa$-approximate ball-optimization oracle for $g$, we can find a solution $\X=(X,r)$ in polynomial time such that $g(\bm{r}(X)) \le \kappa \alpha \OPT_{\I'}$
    where $\I'=(P,F,\delta,\floor{\nicefrac{k}{\alpha}},\LP{\infty},g')$ where $g'(\bm{x'}) = g(\bm{x})$ for all $\bm{x}'\in \mathbb{R}^{k'}_{\ge0}$ and $\bm{x}\in\mathbb{R}^k_{\ge0}$ where $\bm{x}$ is $\bm{x}'$ padded with zeros.
\end{theorem}

By Theorem~5.4 in \cite{chakrabarty-swamy19:norm-k-clustering} it is sufficient to obtain an $(O(1),O(1))$-bicriteria algorithm for \NCCS{\LP{\infty}}{\textsf{MaxOrd}}.
A norm in \textsf{MaxOrd} is described by a set of weight vectors $W\subseteq \nnrvec$.
The value of a \textsf{MaxOrd} norm is defined as $\max_{\bm{w}\in W}\ord{\bm{w}}{\bm{x}}$ for all $\bm{x}\in \nnrvec$.

\begin{Definition}
    For two vectors $\bm{x}\in \mathbb{R}^k_{\ge0}$ and $\bm{x}'\in \mathbb{R}^{k'}_{\ge0}$ and a value $\beta\in \mathbb{N}$ such that $k'\beta \le k$, we say that
    $\bm{x}\le^\beta \bm{x}'$
    if $\bm{x}_i\le \beta\bm{x}'_{\floor{\nicefrac{i}{\beta}}}$ for all $i\in [k]$.
\end{Definition}

\begin{lemma}
    For two vectors $\bm{x}\in \mathbb{R}^k_{\ge 0}$ and $\bm{x}'\in \mathbb{R}^{k'}_{\ge 0}$ and a value $\beta\in \mathbb{N}$ such that $k'\beta \le k$, if it holds that
    $\bm{x}\le^\beta \bm{x}'$,
    then for all vectors $\bm{w}\in \mathbb{R}^k_{\ge0}$ and $\bm{w}'\in \mathbb{R}^{k'}_{\ge0}$ where $\bm{w'}$ is a prefix of $\bm{w}$, it holds that
    \begin{align*}
        \ord{\bm{w}}{\bm{x}} \le \beta^2 \ord{\bm{w}'}{\bm{x}'}
    \end{align*}
\end{lemma}
\begin{proof}
    $\ord{\bm{w}}{\bm{x}} = \sum_{i=1}^k\bm{w}^\downarrow_i\bm{x}^\downarrow_i
         \le  \sum_{i=1}^k\beta \bm{w}'^\downarrow_{\floor{\nicefrac{i}{\beta}}}\bm{x}^\downarrow_i
         \le  \sum_{i=1}^{k'}\beta^2 \bm{w}'^\downarrow_{i}\bm{x}^\downarrow_i$
\end{proof}

\begin{lemma}
    There is a universal constant $\beta >0$(we assume w.l.o.g. that $\beta \in \mathbb{N}$) such that, given an instance $\I=(P,F,\delta,k,\LP{\infty},W)$ and $\eps>0$ , we can compute a solution $\X=(X,r)$ such that 
    \begin{align*}
        \costd{\X}{W} = \max_{\bm{w}\in W}\ord{w}{\bm{r}(X)} \le \beta^2(1+\eps)\OPT_{\I'}
    \end{align*}
    where $\I'=(P,F,\delta,\floor{\nicefrac{k}{\beta}},\LP{\infty},W')$ and $W'$ contains the prefixes of length $\floor{\nicefrac{k}{\beta}}$ from vectors in $W$.
\end{lemma}
\begin{proof}
    Let $k'=\floor{\nicefrac{k}{\beta}}$.
    Then by Lemmas~4.2 and~6.9 from \cite{chakrabarty-swamy19:norm-k-clustering}, we can compute $\Tilde{W}\subseteq \mathbb{R}^{k'}_{\ge0}$, $\bm{t}\in \mathbb{R}^{k'}_{\ge0}$ such that for all $\bm{x}\in \mathbb{R}^{k'}_{\ge0}$ it holds that 

    \begin{itemize}
        \item $\max_{\bm{w}\in W'}\ord{\bm{w}}{\bm{x}}\le \max_{\bm{\tildw}\in \Tilde{W}}\ord{\tildw}{\bm{x}}\le (1+\eps)\max_{\bm{w}\in W'}\ord{\bm{w}}{\bm{x}}$,
        \item $\max_{\bm{\tildw}\in \Tilde{W}}\ord{\tildw}{\bm{t}}\le (1+\eps)\OPT_{\I'}$, and
        \item $\bm{t}^\downarrow \ge \bm{r}^*_{\I'}(X^*_{\I'})^\downarrow$
    \end{itemize}
    Thus, all points in $P$ can be covered by $k'$ balls each upper bounded by a different value of $\bm{t}$.
    Finding such a covering is known as the Non-Uniform k-Center Problem (cf.\ \cite{NUkC}).
    Due to Theorem 1.3 of \cite{NUkC}, we can compute a solution $\X=(X,r)$ covering all points in $P$ with $|X|\le k$ such that $\bm{r}(X) \le^{\beta} \bm{t}$.
    We show that $\X=(X,r)$ is an approximate solution.

    \begin{align*}
        \costd{\X}{W} \le \beta^2 \max_{\bm{w}\in W'}\ord{\bm{w}}{\bm{t}}
        \le  \beta^2 \max_{\bm{\tildw}\in \Tilde{W}}\ord{\bm{\tildw}}{\bm{t}}
        &\le \beta^2 (1+\eps) \OPT_{\I'}
    \end{align*}
    
\end{proof}

\bibliographystyle{plain}
\bibliography{biblio}
\clearpage

\appendix

\section{General Reductions}
In this section we show reductions between special cases of \NNCC{}. 
We achieve this by exploiting properties of ordered norms.

\begin{observation}\label{obs:lonetoord}
    Let $\bm{w} =(w_1,\dots, w_n) \in \nnrvec$ and $W=\sum_{i=1}^nw_i$. Then for all $\bm{x}\in \nnrvec$ it holds that
    \begin{align*}
        w_1\LP{\infty}(\bm{x})\le \ord{\bm{w}}{\bm{x}} \le W \LP{\infty}(\bm{x})
    \end{align*}
    and
    \begin{align*}
          \frac{W}{n}\LP{1}(\bm{x}) \le \ord{\bm{w}}{\bm{x}}\le w_1\LP{1}(\bm{x}) 
    \end{align*}
\end{observation}
The following lemma is a consequence of Lemma 16 in \cite{patton-etal23:submodular-norms} and Theorem 5.4 in \cite{chakrabarty-swamy19:norm-k-clustering}.
\begin{lemma}\label{lem:boundgeneralnorm}
    For any symmetric, monotone norm $f: \nnrvec\rightarrow\nnr$ with an $\kappa$-approximate ball oracle, we can efficiently compute a weight vector $\bm{w}\in \nnrvec$ such that for all $\bm{x}\in \nnrvec$ it holds that
    \begin{align*}
        \ord{\bm{w}}{\bm{x}}\le f(\bm{x}) \le \kappa (3\log n +1)\ord{\bm{w}}{\bm{x}}.
    \end{align*}
\end{lemma}

\begin{lemma}\label{lem:generalred}
    Let $O$ be an arbitrary class for the outer norm and $I$ be an arbitrary class for the inner norm. 
    \begin{enumerate}[(i)]
        \item If there is an $\alpha$-approximation for \NCCS{\LP{1}}{O}, then there is an $\alpha w_1n/W$-approximation for \NCCS{\textnormal{\textnormal{\textsf{Ord}}}}{O}.
        \item If there is an $\alpha$-approximation for \NCCS{I}{\LP{1}}, then there is an $\alpha w_1k/W$-approximation for \NCCS{I}{\textnormal{\textsf{Ord}}}.
        \item If there is an $\alpha$-approximation for \NCCS{\LP{\infty}}{O}, then there is an $\alpha W/w_1$-approximation for \NCCS{\textnormal{\textsf{Ord}}}{O}.
        \item If there is an $\alpha$-approximation for \NCCS{I}{\LP{\infty}}, then there is an $\alpha W/w_1$-approximation for \NCCS{I}{\textnormal{\textsf{Ord}}}.
        \item If there is an $\alpha$-approximation for \NCCS{\LP{\infty}}{O} and an $\beta$-approximation for \NCCS{\LP{1}}{O}, then there is an $\min\{\alpha w_1n/W,\beta W/w_1\}$-approximation for \NCCS{\textnormal{\textsf{Ord}}}{O}.
        \item If there is an $\alpha$-approximation for \NCCS{I}{\LP{\infty}} and an $\beta$-approximation for \NCCS{I}{\LP{1}}, then there is an $\min\{\alpha w_1k/W,\beta W/w_1\}$-approximation for \NCCS{I}{\textnormal{\textsf{Ord}}}.
        \item If there is an $\alpha$-approximation for \NCCS{\LP{\infty}}{O} and an $\beta$-approximation for \NCCS{\LP{1}}{O}, then there is an $\sqrt{n\alpha\beta}$-approximation for \NCCS{\textnormal{\textsf{Ord}}}{O}.
        \item If there is an $\alpha$-approximation for \NCCS{I}{\LP{\infty}} and an $\beta$-approximation for \NCCS{I}{\LP{1}}, then there is an $\sqrt{k\alpha\beta}$-approximation for \NCCS{I}{\textnormal{\textsf{Ord}}}.
        \item If there is an $\alpha$-approximation for \NCCS{\textnormal{\textsf{Ord}}}{O}, then there is an $O(\alpha\log n)$-approximation for \NCCS{\textnormal{\textsf{Sym}}}{O}.
        \item If there is an $\alpha$-approximation for \NCCS{I}{\textnormal{\textsf{Ord}}}, then there is an $O(\alpha\log k)$-approximation for \NCCS{I}{\textnormal{\textsf{Sym}}}. 
    \end{enumerate} 
\end{lemma}
\begin{proof}
    \begin{enumerate}[(i)]
        \item Let $\I=(P,F,\delta,k,\ord{\bm{w}}{\cdot},g)$ be an instance of \NCCS{\textnormal{\textsf{Ord}}}{O} with optimal solution $\X^*_\I=(X^*_\I,\sigma^*_\I)$.
        Let $\I'=(P,F,\delta,k,\LP{1},g)$ be the corresponding instance of \NCCS{\LP{1}}{O} with the $\alpha$-approximate solution $(X,\sigma)$. We use the bounds from \Cref{obs:lonetoord} to obtain the following.

        \begin{align*}
            g((\ord{\bm{w}}{\distv{\sigma}{x}})_{x\in X})& \le g((w_1\LP{1}({\distv{\sigma}{x}}))_{x\in X})\\
            &\le w_1g((\LP{1}({\distv{\sigma}{x}}))_{x\in X})\\
            &\le w_1\alpha g((\LP{1}({\distv{\sigma^*}{x}}))_{x\in X^*})\\
            &\le w_1\alpha g\left(\left(\frac{n}{W}\ord{\bm{w}}{{\distv{\sigma^*}{x}}}\right)_{x\in X^*}\right)\\
            &\le w_1\frac{n}{W}\alpha g\left(\left(\ord{\bm{w}}{{\distv{\sigma^*}{x}}}\right)_{x\in X^*}\right)
        \end{align*}
        \item Let $\I=(P,F,\delta,k,f,\ord{\bm{w}}{\cdot})$ be an instance of \NCCS{I}{\textnormal{\textsf{Ord}}} with optimal solution $\X^*_\I=(X^*_\I,\sigma^*_\I)$.
        Let $\I'=(P,F,\delta,k,f,\LP{1})$ be the corresponding instance of \NCCS{I}{\LP{1}} with the $\alpha$-approximate solution $(X,\sigma)$. We use the bounds from \Cref{obs:lonetoord} to obtain the following.

        \begin{align*}
            \ord{\bm{w}}{f(\distv{\sigma}{x})_{x\in X}}& \le w_1\LP{1}({f(\distv{\sigma}{x})_{x\in X}})\\
            &\le w_1\alpha \LP{1}\left({f(\distv{\sigma^*}{x})_{x\in X^*}}\right)\\
            &\le  w_1\frac{k}{W}\alpha \ord{\bm{w}}{f(\distv{\sigma^*}{x})_{x\in X^*}}
        \end{align*}
        \item 
        Let $\I=(P,F,\delta,k,\ord{\bm{w}}{\cdot},g)$ be an instance of \NCCS{\textnormal{\textsf{Ord}}}{O} with optimal solution $\X^*_\I=(X^*_\I,\sigma^*_\I)$.
        Let $\I'=(P,F,\delta,k,\LP{\infty},g)$ be the corresponding instance of \NCCS{\LP{\infty}}{O} with the $\alpha$-approximate solution $(X,\sigma)$. We use the bounds from \Cref{obs:lonetoord} to obtain the following.
        \begin{align*}
            f((\ord{\bm{w}}{\distv{\sigma}{x}})_{x\in X})& \le f((W\LP{\infty}({\distv{\sigma}{x}}))_{x\in X})\\
            &\le Wf((\LP{\infty}({\distv{\sigma}{x}}))_{x\in X})\\
            &\le W\alpha f((\LP{\infty}({\distv{\sigma^*}{x}}))_{x\in X^*})\\
            &\le W\alpha f\left(\left(\frac{1}{w_1}\ord{\bm{w}}{{\distv{\sigma^*}{x}}}\right)_{x\in X^*}\right)\\
            &\le \frac{W}{w_1}\alpha f\left(\left(\ord{\bm{w}}{{\distv{\sigma^*}{x}}}\right)_{x\in X^*}\right)
            \end{align*}
        \item Let $\I=(P,F,\delta,k,f,\ord{\bm{w}}{\cdot})$ be an instance of \NCCS{I}{\textnormal{\textsf{Ord}}} with optimal solution $\X^*_\I=(X^*_\I,\sigma^*_\I)$.
        Let $\I'=(P,F,\delta,k,f,\LP{\infty})$ be the corresponding instance of \NCCS{I}{\LP{\infty}} with the $\alpha$-approximate solution $(X,\sigma)$. We use the bounds from \Cref{obs:lonetoord} to obtain the following.

        \begin{align*}
            \ord{\bm{w}}{f(\distv{\sigma}{x})_{x\in X}}& \le W\LP{\infty}({f(\distv{\sigma}{x})_{x\in X}})\\
            &\le W\alpha \LP{\infty}\left({f(\distv{\sigma^*}{x})_{x\in X^*}}\right)\\
            &\le  \frac{W}{w_1}\alpha \ord{\bm{w}}{f(\distv{\sigma^*}{x})_{x\in X^*}}
        \end{align*}
        \item Follows by (i) and (iii).
        \item Follows by (ii) and (iv).
        \item Follows by (v) and the fact that $\min\{\alpha w_1n/W,\beta W/w_1\}\le \sqrt{n\alpha\beta}$.
        \item Follows by (vi) and the fact that $\min\{\alpha w_1n/W,\beta W/w_1\}\le \sqrt{k\alpha\beta}$.
        \item Let $\I=(P,F,\delta,k,f,g)$ be an instance of \NCCS{\textnormal{\textsf{Sym}}}{O} with optimal solution $\X^*_\I=(X^*_\I,\sigma^*_\I)$.
        Due to \Cref{lem:boundgeneralnorm} we can efficiently compute a vector $\bm{w}\in \nnrvec$ such that for all vectors $\bm{x}\in \nnrvec$ it holds that
        \begin{align*}
        \ord{\bm{w}}{\bm{x}}\le f(\bm{x}) \le \kappa(3\log n +1)\ord{\bm{w}}{\bm{x}}.
    \end{align*}
        Let $\I'=(P,F,\delta,k,\ord{\bm{w}}{\cdot},g)$ be the corresponding instance of \NCCS{\textsf{Ord}}{O} with the $\alpha$-approximate solution $(X,\sigma)$. 

        \begin{align*}
            g((f(\distv{\sigma}{x}))_{x\in X})& \le g((\kappa(3\log n+1)\ord{\bm{w}
            }{\distv{\sigma}{x}})_{x\in X})\\
            &\le \kappa(3\log n+1) g((\ord{\bm{w}
            }{\distv{\sigma}{x}})_{x\in X})\\
            &\le \kappa(3\log n+1)\alpha g\left((\ord{\bm{w}
            }{\distv{\sigma^*}{x}})_{x\in X^*}\right)\\
            &\le  \kappa(3\log n+1)\alpha g\left((f(\distv{\sigma^*}{x}))_{x\in X^*}\right)
        \end{align*}
        \item Let $\I=(P,F,\delta,k,f,g)$ be an instance of \NCCS{I}{\textnormal{\textsf{Sym}}} with optimal solution $\X^*_\I=(X^*_\I,\sigma^*_\I)$.
        Due to \Cref{lem:boundgeneralnorm} we can efficiently compute a vector $\bm{w}\in \nnr^k$ such that for all vectors $\bm{x}\in \nnr^k$ it holds that
        \begin{align*}
        \ord{\bm{w}}{\bm{x}}\le g(\bm{x}) \le \kappa(3\log k +1)\ord{\bm{w}}{\bm{x}}.
    \end{align*}
        Let $\I'=(P,F,\delta,k,f,\ord{\bm{w}}{\cdot})$ be the corresponding instance of \NCCS{I}{\textsf{Ord}} with the $\alpha$-approximate solution $(X,\sigma)$.  

        \begin{align*}
            g((f(\distv{\sigma}{x}))_{x\in X})& \le \kappa(3\log k+1) \ord{\bm{w}}{(f(\distv{\sigma}{x}))_{x\in X}}\\
            &\le (2\log k+1)\alpha \ord{\bm{w}}{(f(\distv{\sigma^*}{x}))_{x\in X^*}}\\
            &\le (2\log k+1)\alpha g\left((f(\distv{\sigma^*}{x}))_{x\in X^*}\right)
        \end{align*}
    \end{enumerate}
\end{proof}

We conclude a number of results for special cases of \NNCC{} by applying \Cref{lem:generalred}.
\begin{corollary}\label{cor:listofcors}
    Let $W_{in}$ denote the sum of values for the inner ordered norm normalized by the first weight. Furthermore let $W_{out}$ denote the sum of values for the outer ordered norm normalized by the first weight.
    \begin{enumerate}[(i)]
        \item There is an $O(k/W_{out})$-approximation algorithm for \NCCS{\textnormal{\textsf{Top}}}{\textnormal{\textsf{Ord}}}. 
        \item There is an $O(k\log k)$-approximation algorithm for \NCCS{\textnormal{\textsf{Top}}}{\textnormal{\textsf{Sym}}}. 
        \item There is an $O(\log k)$-approximation algorithm for \NCCS{\LP{\infty}}{\textnormal{\textsf{Sym}}}. 
        \item There is an $O(\min\{W_{in},n/W_{in}\})$-approximation algorithm and an $O(\sqrt{n})$-approximation algorithm for \NCCS{\textnormal{\textsf{Ord}}}{\LP{1}}. 
        \item There is an $O(\min\{nk/(W_{in}W_{out}),W_{in}\}$-approximation algorithm and an $O(\sqrt{nk})$-approximation algorithm for \NCCS{\textnormal{\textsf{Ord}}}{\textnormal{\textsf{Ord}}}.
        \item There is an $O(\min\{nk/W_{in},W_{in}\}\log k)$-approximation algorithm and an $O(\sqrt{nk}\log k)$-approximation algorithm for \NCCS{\textnormal{\textsf{Ord}}}{\textnormal{\textsf{Sym}}}.
        \item There is an $O(\sqrt{n}\log n)$-approximation algorithm for \NCCS{\textnormal{\textsf{Sym}}}{\LP{1}}. 
        \item There is an $O(\sqrt{nk}\log n)$-approximation algorithm for \NCCS{\textnormal{\textsf{Sym}}}{\textnormal{\textsf{Ord}}}. 
        \item There is an $O(\sqrt{nk}\log n \log k)$-approximation algorithm for \NCCS{\textnormal{\textsf{Sym}}}{\textnormal{\textsf{Sym}}}. 
    \end{enumerate}
\end{corollary}
\begin{proof}
    \begin{enumerate}[(i)]
        \item Follows by the $O(1)$-approximation algorithm for \NCCS{\textnormal{\textsf{Top}}}{\LP{1}} (\Cref{thm:apxtoplone}) and \Cref{lem:generalred}.
        \item Follows by (i) and \Cref{lem:generalred}.
        \item Follows by the $O(1)$-approximation algorithm for \NCCS{\LP{\infty}}{\textnormal{\textsf{Ord}}}(\Cref{thm:linford} and \Cref{lem:generalred}.
        \item Follows by the $O(1)$-approximation algorithms for \kmed{} and \msr{} and \Cref{lem:generalred}.
        \item Follows by the $O(k/W_{out})$-approximation algorithm for \NCCS{\LP{1}}{\textnormal{\textsf{Ord}}}(i), the $O(1)$-approximation algorithm for \NCCS{\LP{\infty}}{\textnormal{\textsf{Ord}}}(\Cref{thm:linford}) and \Cref{lem:generalred}.
        \item Follows by (v) and \Cref{lem:generalred}.
        \item Follows by (iv) and \Cref{lem:generalred}.
        \item Follows by (v) and \Cref{lem:generalred}.
        \item Follows by (vi) and \Cref{lem:generalred}.
    \end{enumerate}
\end{proof}

\section{\texorpdfstring{Proof of \Cref{lem:bipoint}}{Proof of Lemma 21}}
\begin{proof}[Proof of \Cref{lem:bipoint}]
    Ideally, we would want to find some $\lambda$ such that our approximation algorithm for the facility location version of Ball $k$-median opens exactly $k$ facilities.
    As we cannot guarantee that, we settle for some $\lambda_1,\lambda_2$ such that $|\lambda_1-\lambda_2|$ is small (more precisely, $|\lambda_1-\lambda_2| \le (\varepsilon \delta_{min})(3|F|)$, where $\delta_{min}$ is the minimum non-zero distance between a facility and a client), our approximation for the facility location version of Ball $k$-median when the opening cost is $\lambda_2$ opens more than $k$ facilities, and our approximation for the facility location version of \PZ{} when the opening cost is $\lambda_1$ opens at most $k$ facilities.
    
    We start with $\lambda_1= |P| \delta_{max}$ 
    (where $\delta_{max}$ is the maximum distance between a facility and a client);
    whatever solution $\X_1'=(X'_1,r'_1)$ we get from our approximation algorithm, we convert our solution $\X_1'$ to a solution $\X_1=(X_1,r_1)$ by only keeping the (less than $k$) facilities in $T$, and closing the rest.
    The connection cost can increase by at most $|P| \delta_{max}$,
    which is at most as large as the decrease in the opening cost.
    Therefore $\costz{\X_1} + 3\lambda_1 |X_1| \le \costz{\X_1'}+ 3\lambda_1 |X_1'| \le 3 \OPT_\I + 3\lambda_1 k$.
    
    Similarly for $\lambda_2=0$, we convert our solution $\X_2'=(X_2',r_2')$ to a solution $\X_2=(X_2,r_2)$ opening more than $k$ facilities by opening facilities of zero radius.
    We had $\costz{\X_2'} + 3\lambda_2 |X_2'| \le 3\OPT_\I + 3\lambda_2 k$.
    Notice that $X_2$ contains all the facilities of $X_2'$ therefore $\costz{\X_2} \le \costz{\X_2'}$.
    Also $3\lambda_2 |X_2'| = 3\lambda_2 |X_2| = 0$ because $\lambda_2=0$.
    Therefore $\costz{\X_2} + 3\lambda_2 |X_2| \le \OPT_\I + 3\lambda_2 k$.

    As both $X_1$ and $X_2$ contain $T$, $X_1$ has less than $k$ facilities, and $X_2$ has more.
    Furthermore, the only facilities we open that were not suggested by our approximation algorithm are of zero radius.
    Therefore it is also true that if $x \in (X_1\cup X_2) \setminus (X_1\cap X_2)$ then $\ell\cdot r_1(x) < \eps\OPT_\I$ and $\ell\cdot r_2(x) < \eps\OPT_\I$.

    We now perform a binary search with $\lambda \in [0,|P|\delta_{max}]$: we continue on the bottom half of the search space when our approximation algorithm with $\lambda$ being the middle point of the search space returns a solution $\X=(X,r)$ with $|X| \le k$ (and setting $\X_1=\X$), or continue to the top half and setting $\X_2=\X$ otherwise.
    We stop when the search space is $[\lambda_2,\lambda_1]$ and $|\lambda_1-\lambda_2| <(\varepsilon \delta_{min})/(3|F|)$.    
    By \Cref{lem:lmpalgo} we get that $T\subseteq X_1\cap X_2$ and $r_1(x) = r^*_\I(x)=r_2(x)$ for all $x\in T$.
    Again due to \Cref{lem:lmpalgo}, it follows that $\ell\cdot3 r_1(x) \le \eps \OPT_\I$ and $\ell\cdot3 r_2(x) \le \eps \OPT_\I$ for $x\in (X_1\cup X_2) \setminus T$. 
    The binary search gives us that $X_1$ has at most $k$ facilities, and $X_2$ has more.

    We now have:
    \begin{align*}
        \costz{\X_1} + 3\lambda_1 |X_1| &\le 3\OPT_\I + 3\lambda_1 k
    \end{align*}
    and
    \begin{align*}
        \costz{\X_2} + 3\lambda_2 |X_2| &\le 3\OPT_\I + 3\lambda_2 k \implies\\
        \costz{\X_2} + 3\lambda_1 |X_2| &\le 3\OPT_\I + 3\lambda_1 k + 3|\lambda_1-\lambda_2||X_2| \le 3\OPT_\I + 3\lambda_1 k + \varepsilon \OPT_\I  
    \end{align*}
    
   Let $a,b$ be the convex combination such that $a|X_1|+b|X_2| = k$.
   Multiplying the first inequality by $b$, the second by $a$, and adding them together gives
   \begin{align*}
       &a\cdot \costz{\X_1} + b\cdot \costz{\X_2} + 3\lambda_1 (a|X_1|+b|X_2|) \le (a+b)3\OPT_\I + 3(a+b)\lambda_1 k + b\cdot \varepsilon \OPT_\I \implies \\
       &a\cdot \costz{\X_1} + b\cdot \costz{\X_2} \le (3+\eps)\OPT_\I
   \end{align*}
\end{proof}
\end{document}